\documentclass[12pt]{article}

\usepackage{amssymb, amsmath, amsthm}
\usepackage[left=2cm,top=4cm,bottom=2.5cm,right=2cm]{geometry}
\usepackage[utf8]{inputenc} % allow utf-8 input
\usepackage[T1]{fontenc}    % use 8-bit T1 fonts
\usepackage{hyperref}       % hyperlinks
\usepackage{url}            % simple URL typesetting
\usepackage{booktabs}       % professional-quality tables
\usepackage{amsfonts}       % blackboard math symbols
\usepackage{nicefrac}       % compact symbols for 1/2, etc.
\usepackage{microtype}      % microtypography
\usepackage{xcolor}         % colors
\usepackage{physics}        % Bra-Ket Notation
\usepackage{amssymb, amsmath, amsthm}
\usepackage{xpatch}
\usepackage[pdftex]{graphicx}
\usepackage{authblk}
\usepackage[symbol]{footmisc}
\usepackage{multirow}       % For multirows in tabulars
\usepackage{diagbox}       % For top left entry in tabular
\usepackage{makecell}
\usepackage{caption}
\usepackage{tikz}
\usepackage{ragged2e}
\usetikzlibrary{arrows.meta, positioning}
\makeatletter
\AtBeginDocument{\xpatchcmd{\@thm}{\thm@headpunct{.}}{\thm@headpunct{}}{}{}}
\makeatother

\newcommand{\mathbbm}[1]{\text{\usefont{U}{bbm}{m}{n}#1}}
\DeclareMathOperator\arctanh{arctanh}

\newtheorem{theorem}{Theorem}
\newtheorem*{theorem*}{Theorem}
\newtheorem{proposition}[theorem]{Proposition}
\newtheorem{corollary}[theorem]{Corollary}
\newtheorem{lemma}{Lemma}
\newtheorem*{lemma*}{Lemma}
\newtheorem{definition}{Definition}
\newtheorem*{definition*}{Definition}
\newtheorem{remark}{Remark}
\newtheorem*{remark*}{Remark}
\newtheorem*{blank}{\hspace*{-5pt}}

\newtheorem*{fact*}{Fact}

\DeclareMathOperator{\Var}{Var}

\begin{document}

\title{Sample-optimal learning of quantum states using gentle measurements}
\date{}

\author[1]{\small Cristina Butucea \footnote{\texttt{cristina.butucea@ensae.fr}}}
\author[2]{\small Jan Johannes \footnote{\texttt{johannes@math.uni-heidelberg.de}}}
\author[1, 2]{\small Henning Stein \footnote{\texttt{henning.stein@math.uni-heidelberg.de}}}

\affil[1]{\footnotesize CREST, ENSAE, Institut Polytechnique de Paris, 91120 Palaiseau, France}
\affil[2]{\footnotesize Heidelberg University, 69120 Heidelberg, Germany}

\maketitle

\begin{abstract}
  Gentle measurements of quantum states do not entirely collapse the initial state. Instead, they provide a post-measurement state at a prescribed trace distance $\alpha$ from the initial state together with a random variable used for quantum learning of the initial state. We introduce here the class of $\alpha-$locally-gentle measurements ($\alpha-$LGM) on a finite dimensional quantum system which are product measurements on product states and prove a strong quantum Data-Processing Inequality (qDPI) on this class using an improved relation between gentleness and quantum differential privacy. We further show a gentle quantum Neyman-Pearson lemma which implies that our qDPI is asymptotically optimal (for small $\alpha$). This inequality is employed to show that the necessary number of quantum states for prescribed accuracy $\epsilon$ is of order $1/(\epsilon^2 \alpha^2)$ for both quantum tomography and quantum state certification. Finally, we propose an $\alpha-$LGM called quantum Label Switch that attains these bounds. It is a general implementable method to turn any two-outcome measurement into an $\alpha-$LGM. 
\end{abstract}

\section{Introduction}
\label{sec::Introduction}

Quantum measurements, which are physical devices to extract information about the quantum state or system of states at hand, are inherently destructive. This destruction is even reflected in the widely used term \textit{collapse} of a wave function. Quantum instruments are physical manipulations that can produce both statistical information (a random variable is observed) and a resulting collapsed quantum state. 

Recently, the pioneering work by \cite{aaronson2019gentle} formalized a class of measurements called \textit{gentle} which produce both statistical information and a resulting state which is at a prescribed distance from the measured state. We call the resulting state a "post-measurement state" rather than a "collapsed state" to highlight that it is indeed not far away from the initial state. While some principles of gentle (or weak) measurements were already known to researchers \cite{Bennett_2006, RevModPhys.52.341, hayashi2002simpleconstructionquantumuniversal} the introduced formalization sparked a renewed interest in the problem. In recent times, gentle measurements have been studied in a variety of contexts such as quantum event learning \cite{watts2024quantumeventlearninggentle}, Quantum Information Leakage \cite{10886009}, Quantum encryption \cite{Gulbahar2023}, quantum machine learning \cite{abbas2023on}, tests for group isomorphisms \cite{harrow_sequential_2017} and even the foundations of quantum theory \cite{wakakuwa2021gentlemeasurementprinciplequantum}. In each of these cases, the properties of gentle measurement prove essential for the tasks studied. The upside of gentleness is that after measuring a state gently the resulting post-measurement state is still of practical value. Rather than collapsing into a completely unrelated state the post-measurement state for a gentle measurement is by definition close to the original one. As such, it may be used again to store or to extract further information from it. These properties have already been observed in physical experiments \cite{katz_reversal_2008} showing the physical feasibility of gentle measurement. For study of their statistical properties, especially the connection between gentleness and quantum Differential Privacy (qDP) established in \cite{aaronson2019gentle} is proving fruitful for us in the development of new results in both fields. 

While \cite{aaronson2019gentle} only consider gentleness on a whole system, that we now call \textit{global gentleness}, we give a new definition of \textit{local gentleness} and \textit{locally-gentle measurements}. In contrast to globally-gentle measurements their local counterparts are unentangled and measure each subsystem of a quantum system one at a time. This approach allows for a far more physically feasible description of measurements \cite{Bubeck_9318006, liu2024rolesharedrandomnessquantum} as holding and manipulating a large number of quantum states at one time surpasses current technological capabilities.  

With the ever rising capabilities of QPUs and accessibility to qubits we believe that it is important to study the properties of gentle measurements for those types of states. In a recent paper \cite{abbas2023on} discussed how gentle measurements may be necessary in order to implement quantum backpropagation algorithms for machine learning. Given the success of classical backpropagation algorithms, studying gentle measurements for qubits, which constitute the building blocks of any modern quantum computer, is a integral part in the study of advanced quantum machine learning algorithms. 

In this paper, we consider the problems of estimating and testing qubits using gentle measurements. We start by giving a brief introduction into the necessary concepts of quantum and classical statistics in Section \ref{sec::background}. In Section \ref{sec::gentleness} we then introduce the concepts of gentle measurements. We state some basic properties of gentle measurements and improve the constant relating gentleness and differential privacy established by \cite{aaronson2019gentle} and extend its applicability to a wider range of values for the gentleness parameter $\alpha$. Furthermore, we give a short result on how to construct gentle measurements. In Section \ref{sec::qDPI} we then prove a proper quantum data-processing inequality (qDPI) for gentle measurements that relates the information extracted from a state and the destruction of that state by an arbitrary gentle measurement similar to a classical result for differentially-private mechanisms due to \cite{duchi2014localprivacydataprocessing}. We demonstrate the power of our qDPI by deriving theoretical guarantees for both locally-gentle quantum state certification and locally-gentle quantum tomography in Section~\ref{Sec:Nec}. Moreover, in Section \ref{sec::quantum_label_switch}, we develop a quantum locally-gentle measurement technique that we call the \textit{quantum Label Switch} (qLS), implementable using well known tools of quantum computing, that allows for locally-gentle state certification and locally-gentle tomography for qubits. Given a gentleness parameter that the user can fix in advance, we build explicit procedures for both problems that use the qLS measurement and show that a loss of rate is unavoidable in both problems due to the gentleness constraints. Finally, some further proofs can be found in Section \ref{sec::additioinal_proofs_and_calculations}.

\section{Quantum statistics background}
\label{sec::background}
Let us recall some basic definitions of both quantum and classical statistics. We will use the typical physicist notation where $\ket{\psi} \in \mathcal{H}$ denotes a unit vector in a Hilbert space $\mathcal{H}$ and $\bra{\varphi} \in \mathcal{H}^* \cong \mathcal{H}$ its dual vector. Consequently, the inner product between $\ket{\varphi}$ and $\ket{\psi}$ can be written as $\bra{\varphi}\ket{\psi}$ and the projection onto the space spanned by $\ket{\psi}$ is denoted as $\ket{\psi}\bra{\psi}$.

\subsection{States and measurements}
The laws of quantum mechanics dictate that a quantum system is described by a Hilbert space $\mathcal{H}$ and that the state of the system is given by a unit vector $\ket{\psi} \in \mathcal{H}$. \footnote{We will not go into the exact origin of $\ket{\psi}$ other than it being a solution to the famous "Schrödinger equation"} Even if we had complete knowledge of the state, any value we measure would still be probabilistic. Furthermore, when we have incomplete knowledge of the system, we must deal with the additional uncertainty of which state we are dealing with. From the outside, we are dealing with a probability distribution over a number of possible pure states $\ket{\psi_i}$. It turns out that the correct way to describe this probability distribution over pure states is as an operator $\rho = \sum_{i} \lambda_i \ket{\psi_i}\bra{\psi_i}$ for orthogonal $\ket{\psi_i}$ such that $\sum_i \lambda_i = 1$. This motivates the following definition.

\begin{definition}
\label{defn::quantum_system}
    A finite dimensional quantum system is described by the space $\mathcal{H} = \mathbb{C}^d$. The quantum state of the system is given by a positive semi-definite operator $\rho \geq 0$ on $\mathbb{C}^{d }$ with $\Tr(\rho) = 1$.
\end{definition}
A state $\rho$ is pure if and only if it has rank 1 and can be written as $\rho = \ket{\psi}\bra{\psi}$. In that case we identify the state $\rho$ by its representational (normed) vector $\ket{\psi}$. Note that any two normed vectors that only differ by a phase represent the same pure quantum state. That is, if there exists $\lambda \in \mathbb{C}$ with $|\lambda| = 1$ such that $\ket{\psi} = \lambda \ket{\psi'}$, then $\ket{\psi}$ and $\ket{\psi'}$ represent the same pure state. If a state is not pure, it is said to be mixed and can be written as $\rho = \sum_i \lambda_i \ket{\psi_i}\bra{\psi_i}$ for an orthonormal basis $\ket{\psi_i}$ of $\mathbb{C}^d$ and $\lambda_i \geq 0$ with $\sum_i \lambda_i = 1$. We denote by $\mathcal{S}(\mathbb{\mathcal{H}})$ the set of quantum states on $\mathcal{H}$ and by $\mathcal{S}_{pure}(\mathbb{\mathcal{H}})$ the set of pure quantum states on $\mathcal{H}$. For two states $\rho_1, \rho_2 $ on $\mathbb{C}^d$, their product is given by $\rho_1 \otimes \rho_2 $ on the product space $\mathbb{C}^d \otimes \mathbb{C}^d$, where $\otimes$ denotes the tensor product. For product states we use the shorthand notation $\rho^{\otimes n} $ on $(\mathbb{C}^d)^{\otimes n}$ to denote the $n-$fold tensor product. The state of a system cannot be observed directly but only through quantum measurements. There are several different notions of quantum measurements. We use the definition in \cite{NielsenChuang}.

\begin{definition}
\label{defn::quantum_measurement}
    A measurement $M$ on a quantum system with outcomes in a measurable space $(\mathcal{Y}, \mathcal{P}(\mathcal{Y}))$, for countable $\mathcal{Y}$, is a set of measurement operators $M_y $ in $\mathbb{C}^{d \times d}, y \in \mathcal{Y}$ such that 
    $$
    \sum_{y\in \mathcal{Y}} M_y^*M_y = \mathbbm{1}.
    $$
    Here, $\mathbbm{1} := \mathbbm{1}_{\mathcal{H}}$ denotes the identity matrix on $\mathcal{H}$.
\end{definition}
    The last equation is called the completeness relation. Measuring the state $\rho$ with the measurement $M$ yields the random variable $R_M$ having outcome $y$ with probability 
    $$
    p_{\rho}^{R_M}(y) := \mathcal{P}_{\rho}(R_M = y) = \Tr(\rho M_y^*M_y)
    $$ and the measurement alters the system, such that after the measurement, the state is given by 
    $$
    \rho_{M \to y} = \frac{1}{\mathcal{P}_{\rho}(R_M = y)} M_y \rho M_y^*.
    $$

If the state $\rho = \ket{\psi} \bra{\psi}$ is pure, the outcome probability simplifies to 
$$
\mathcal{P}_{\rho}(R_M = y) = \bra{\psi} M_y^* M_y \ket{\psi}.
$$
In that case the post-measurement state is also pure and given by 
$$
\ket{\psi}_{M \to y} = \frac{1}{\sqrt{\mathcal{P}_{\rho}(R_M = y)}} M_y \ket{\psi}.
$$ 

Other definitions for measurements include "projector values measurements" (PVM) where the defining measurement operators $M_y$ have to be projectors and "positive operator values measurements" (POVM). A POVM on a measurable space $(\Omega, \Sigma)$ is a map $M: \Sigma \to \mathcal{H}$ with
\begin{itemize}
    \item[(i)] $M(E) \geq 0$ for all measurable sets $E $ in $\Sigma$.

    \item[(ii)] $M\left(\bigcup_{i = 1}^{\infty} E_i \right) = \sum_{i = 1}^{\infty} M(E_i)$ for all disjoint measurable sets $E_i$ in $\Sigma$.

    \item[(iii)] $M(\Omega) = \mathbbm{1}$.
\end{itemize}
The properties (i)-(iii) again make sure that for every $\rho \in \mathcal{S}(\mathcal{H})$ the mapping
\begin{equation*}
    E \to \mathcal{P}_{\rho}(E) = \Tr(\rho M(E))
\end{equation*}
defines a probability measure on $(\Omega, \Sigma)$. While POVMs seem to define a more general class than Definition~\ref{defn::quantum_measurement}, they lack information about the post-measurement state which is a crucial aspect of gentleness if the outcome space $\Omega$ is uncountable. In the case of a countable $\Omega$ we can set $\mathcal{Y} = \Omega$ and $M_y = \sqrt{M(\{y\})}$, where $\sqrt{M}$ denotes the matrix square root of a positive semi-definite $M$. The so-defined operators $M_y$ then define a quantum measurement in the sense of Definition~\ref{defn::quantum_measurement}. Conversely, for a measurement $M = (M_y)_{y \in \mathcal{Y}}$ the mapping defined by $M(E) = \sum_{y \in E} M_y^*M_y$ defines a POVM on $(\mathcal{Y}, \mathcal{P}(\mathcal{Y}))$. As such, Definition~\ref{defn::quantum_measurement} is the most general definition that still allows for the the description of post-measurement states which are essential for gentleness. PVMs, while often used in standard quantum procedures, do not capture the full range of possible measurements as all non-trivial PVMs are not full rank. We will however see that gentle measurements need to be full rank excluding PVMs from the list of possible gentle measurements. The Naimark-Theorem \cite{doran1986characterizations} makes a connection between PVMs and general measurements showing that every general measurement may be implemented as a PVM on a larger Hilbert space. This duality can be seen in the description of the quantum Label Switch mechanism for qubits. While explicit formulas, see \eqref{eqn::label_switch_mathematical_formulation}, give the mathematical description of the qLS in the sense of Definition \ref{defn::quantum_measurement}, the physical realization of this mechanism includes pairing the qubits with an ancillary qubit that will be measured by a PVM.

Similarly to product states, we can consider product measurements. For two quantum mechanical measurements $M^{(1)}$ on $\mathcal{H}_1$ with measurement operators $(M_{y_1}^{(1)})_{y_1 \in \mathcal{Y}_1}$ and $M^{(2)}$ on $\mathcal{H}_2$ with measurement operators $(M_{y_2}^{(2)})_{y_2 \in \mathcal{Y}_2}$ we can consider the product measurement $M^{(1)} \otimes M^{(2)}$ given by the measurement operators
\begin{equation*}
    \left(M_{y_1}^{(1)} \otimes M_{y_2}^{(2)} \right)_{y_1 \in \mathcal{Y}_1, y_2 \in \mathcal{Y}_2}.
\end{equation*}
To see that this set of operators defines a measurement we have to check its completeness relation
\begin{align*}
    \sum_{(y_1, y_2) \in \mathcal{Y}_1 \times \mathcal{Y}_2} (M_{y_1}^{(1)} \otimes M_{y_2}^{(2)})^* (M_{y_1}^{(1)} \otimes M_{y_2}^{(2)}) &= \sum_{(y_1, y_2) \in \mathcal{Y}_1 \times \mathcal{Y}_2} ({M_{y_1}^{(1)}}^* \otimes {M_{y_2}^{(2)}}^*) (M_{y_1}^{(1)} \otimes M_{y_2}^{(2)})
    \\
    &= \sum_{(y_1, y_2) \in \mathcal{Y}_1 \times \mathcal{Y}_2} ({M_{y_1}^{(1)}}^* M_{y_1}^{(1)}\otimes {M_{y_2}^{(2)}}^* M_{y_2}^{(2)})
    \\
    &= \sum_{y_1 \in \mathcal{Y}_1} \sum_{y_2 \in \mathcal{Y}_2} ({M_{y_1}^{(1)}}^* M_{y_1}^{(1)}\otimes {M_{y_2}^{(2)}}^* M_{y_2}^{(2)})
    \\
    &= \sum_{y_1 \in \mathcal{Y}_1} \left({M_{y_1}^{(1)}}^* M_{y_1}^{(1)}\otimes  \sum_{y_2 \in \mathcal{Y}_2} {M_{y_2}^{(2)}}^* M_{y_2}^{(2)}\right)
    \\
    &= \mathbbm{1}_{\mathcal{H}_1} \otimes \mathbbm{1}_{\mathcal{H}_2} = \mathbbm{1}_{\mathcal{H}_1 \otimes \mathcal{H}_2}.
\end{align*}
As such, $M^{(1)} \otimes M^{(2)}$ indeed defines a measurement on $\mathcal{H}_1 \otimes \mathcal{H}_2$. Measuring the state $\rho_1 \otimes \rho_2$ with the measurement $M^{(1)} \otimes M^{(2)}$ yields two independent random variables $R_{M^{(1)}}$ and $R_{M^{(2)}}$ on $\mathcal{Y}_1 \times \mathcal{Y}_2$. 

Examples of non-product measurements are the so-called coherent measurements. These are measurements that may act on the whole product space \textit{coherently}, i.e. acting on all registers at once. Coherent measurements are often needed to obtain optimality, for example in optimally discriminating two quantum states (\cite{NussbaumSzkola2009}) but are physically very difficult to implement as they require access to the whole product system at once at the time of measuring. The case of a product measurement on a non-product state is a physically less demanding task. Similarly to the definition of $n$-fold product states, we write $M^{\otimes n}$ to denote $n$ measurements $M$ acting on each register independently.

\subsection{Qubits}
Qubits are a special case of quantum states for $d = 2$. The name arises from merging the words "quantum bit" into one. In quantum computing they play the analogous role of classical bits in classical computing. Every qubit can be described using three real-valued parameters $\mathbbm{r}_x, \mathbbm{r}_y, \mathbbm{r}_z$. As such, the set of qubits forms a three dimensional manifold, known as the Bloch sphere. Let us further develop this. To do this, we fix an orthonormal basis $\ket{0}, \ket{1}$ of $\mathbb{C}^2$. Although it is not important which basis vectors we chose, it is most common and useful to think of
\begin{equation*}
    \ket{0} = \begin{bmatrix}
        1 
        \\
        0
    \end{bmatrix} \hspace{20pt} \text{and} \hspace{20pt} \ket{1} = \begin{bmatrix}
        0 
        \\
        1
    \end{bmatrix}.
\end{equation*}
Based on the two basis vectors $\ket{0}$ and $\ket{1}$, we define six vectors which will represent the axis directions in the Bloch sphere visualization.
\begin{alignat*}{3}
    \ket{{x_+}} &:= \frac{1}{\sqrt{2}} \left( \ket{0} + \ket{1} \right), \hspace{20pt} &&\ket{{x_-}} &&:= \frac{1}{\sqrt{2}} \left( \ket{0} - \ket{1} \right)
    \\
    \ket{{y_+}} &:= \frac{1}{\sqrt{2}} \left( \ket{0} + i\ket{1} \right), && \ket{{y_-}} &&:= \frac{1}{\sqrt{2}} \left( \ket{0} - i\ket{1} \right)
    \\
    \ket{{z_+}} &:= \ket{0}, \phantom{\frac{1}{\sqrt{2}}} && \ket{{z_-}} &&:= \ket{1}.
\end{alignat*}
Note that the sets $\{ \ket{{x_+}}, \ket{{x_-}} \}$, $\{ \ket{{y_+}}, \ket{{y_-}} \}$, $\{ \ket{{z_+}}, \ket{{z_-}} \}$ each form an orthonormal basis of $\mathbb{C}^2$ and it holds
\begin{equation*}
\label{eqn::mutually_unbiasedness}
    |\bra{x_{\pm}}\ket{y_{\pm}}|^2 = \frac{1}{2}, \hspace{10pt} |\bra{x_{\pm}}\ket{z_{\pm}}|^2 = \frac{1}{2}, \hspace{10pt} |\bra{y_{\pm}}\ket{z_{\pm}}|^2 = \frac{1}{2}.
\end{equation*}
Such a collection of vectors are called mutually unbiased bases (MUBs). We can use these bases to further define the so called Pauli-basis. This is a basis of the hermitian $2 \times 2$-matrices composed of $\mathbbm{1}, \sigma_x, \sigma_y, \sigma_z$, where
\begin{align*}
    \sigma_x &= \begin{bmatrix}
        0 & 1
        \\
        1 & 0
    \end{bmatrix} = \ket{{x_+}}\bra{{x_+}} - \ket{{x_-}}\bra{{x_-}}
    \\
    \sigma_y &= \begin{bmatrix}
        0 & -i
        \\
        i & 0
    \end{bmatrix} = \ket{{y_+}}\bra{{y_+}} - \ket{{y_-}}\bra{{y_-}}
    \\
    \sigma_z &= \begin{bmatrix}
        1 & 0
        \\
        0 & -1
    \end{bmatrix} = \ket{{z_+}}\bra{{z_+}} - \ket{{z_-}}\bra{{z_-}}.
\end{align*}
The Pauli-basis representation of a qubit $\rho$ is then given as 
\begin{equation}
\label{eqn::qubit_in_pauli_basis}
    \rho = \rho(\mathbbm{r}) = \frac{1}{2} \left( \mathbbm{1} + \mathbbm{r}_x \sigma_x + \mathbbm{r}_y \sigma_y + \mathbbm{r}_z \sigma_z \right).
\end{equation}
The eigenvalues of the state $\rho(\mathbbm{r})$ are given by $\frac{1}{2}(1 \pm \norm{\mathbbm{r}}_2)$ (see Lemma~\ref{lem::trace_norm_qubits}). As such, in order for $\rho(\mathbbm{r})$ to be a valid quantum state, \emph{i.e.} positive and self-adjoint, it must hold $\norm{\mathbbm{r}}_2 \leq 1$. Indeed, in this way, every qubit can be uniquely identified with a point in the three dimensional real ball with radius 1. Furthermore, for $\norm{\mathbbm{r}}_2 = 1$, the eigenvalues of $\rho(\mathbbm{r})$ are $0$ and $1$. As such, the pure qubits correspond exactly to the points on boundary of the ball, that being the three-dimensional unit sphere.
 
\begin{figure}[ht]
    \begin{minipage}[c]{100pt}
    \end{minipage}\hfill
    \begin{minipage}[c]{0.3\textwidth}
        \fbox{\includegraphics[width=.9\linewidth]{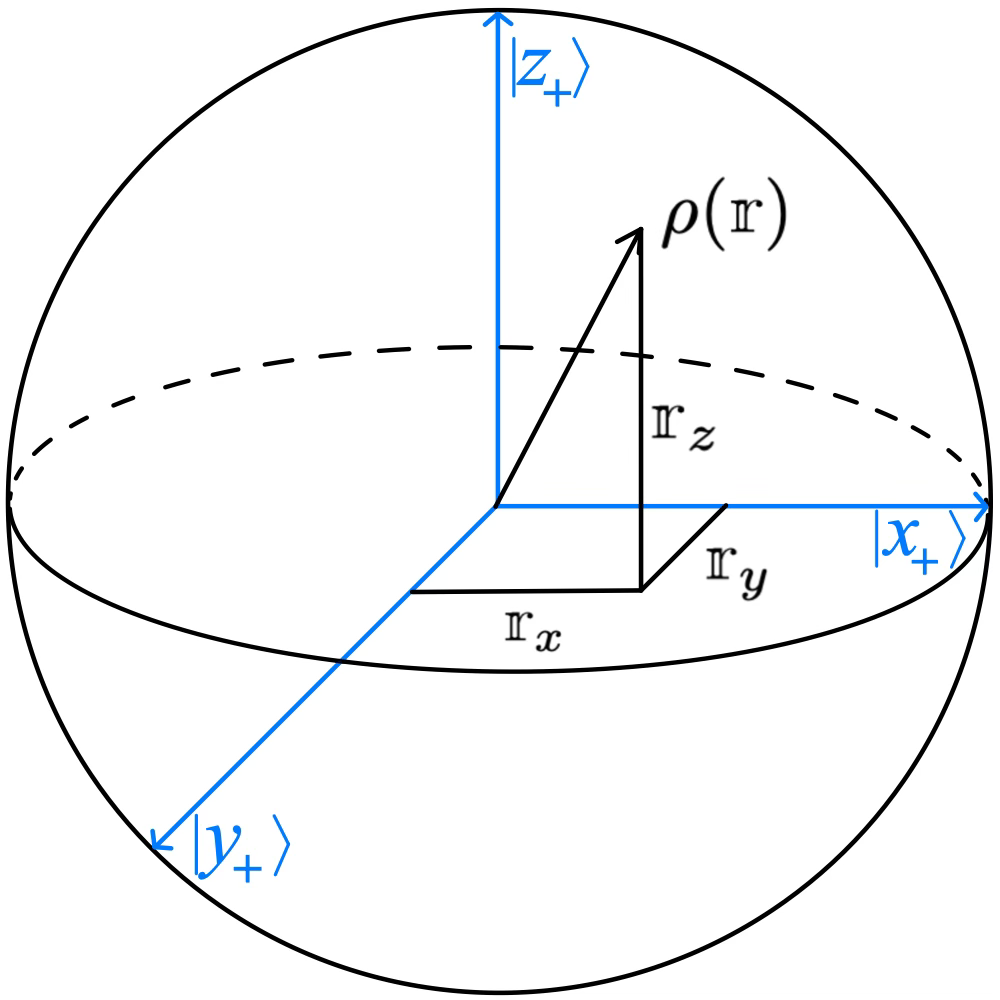}}
    \end{minipage}\hfill
    \begin{minipage}[c]{270pt}
        \caption{Bloch sphere. Every qubit can be represented uniquely by a point in the ball of radius 1.}
    \end{minipage}
\end{figure}
 
\subsection{Distance measures of quantum states}
In order to quantify the closeness of two quantum states, we need to define distance measures on the space of quantum states. One of the most common distance measures for quantum states is the trace-norm distance.
\begin{definition}
    Let $\rho_1, \rho_2 \in \mathcal{S}(\mathbb{C}^d)$ be two quantum states. The trace-norm between the two states is given by 
    $$
    \norm{\rho_1 - \rho_2}_{Tr} = \frac{1}{2} \Tr(|\rho_1 - \rho_2|) = \frac{1}{2} \sum_i |\lambda_i|,
    $$
    where the $\lambda_i$ denote the eigenvalues of $\rho_1-\rho_2$. 
\end{definition}
The calculation of the trace-norm is difficult in general as we need to calculate the spectral decomposition of $\rho_1 - \rho_2$ but simplifies a lot if $\rho_1$ and $\rho_2$ are pure states.

\begin{lemma}{\normalfont (\cite{kargin2003chernoffboundefficiencyquantum})} \label{lemmaTraceNorm}
    Let $\rho_1 = \ket{\psi_1}\bra{\psi_1}, \rho_2 = \ket{\psi_2}\bra{\psi_2}$ be two pure states. Then the trace-norm distance between the two is given by 
    \begin{equation*}
        \norm{\rho_1 - \rho_2}_{Tr} = \sqrt{1 - |\bra{\psi_1}\ket{\psi_2}|^2}.
    \end{equation*}
    For pure product states $\rho_1^{\otimes n}, \rho_2^{\otimes n}$ the trace-norm is given by
    \begin{equation*}
        \norm{\rho_1^{\otimes n} - \rho_2^{\otimes n}}_{Tr} = \sqrt{1 - |\bra{\psi_1}\ket{\psi_2}|^{2n}} = \sqrt{1 - (1 - \norm{\rho_1 - \rho_2}_{Tr}^2)^n}.
    \end{equation*}
\end{lemma}

The trace-norm of pure qubits has the nice property of being related to the euclidian distance of the parameter vectors in the Bloch-sphere/Pauli-basis representation. We will formalize this result in the following lemma.
\begin{lemma}
\label{lem::trace_norm_qubits}
Let $\rho(\mathbbm{r}_1), \rho(\mathbbm{r}_2) \in \mathcal{S}(\mathbb{C}^2)$. Then, for $i = 1,2$, the eigenvalues of $\rho(\mathbbm{r}_i)$ are given by $\frac{1}{2}(1 \pm \norm{\mathbbm{r}_i}_2)$ and
\begin{equation*}
    \norm{\rho(\mathbbm{r}_1) - \rho(\mathbbm{r}_2)}_{Tr} = \frac{1}{2} \norm{\mathbbm{r}_1 - \mathbbm{r}_2}_2.
\end{equation*}
\end{lemma}

\begin{proof}
    For any $\mathbbm{r} \in \mathbb{R}^3$ let
    \begin{equation*}
        \mathbbm{r} \cdot \sigma := \mathbbm{r}_x \sigma_x + \mathbbm{r}_y \sigma_y + \mathbbm{r}_z \sigma_z = \begin{bmatrix}
            \mathbbm{r}_z & \mathbbm{r}_x - i \mathbbm{r}_y
            \\
            \mathbbm{r}_x + i \mathbbm{r}_y & - \mathbbm{r}_z
        \end{bmatrix}
    \end{equation*}
    Then we have $\Tr\left( \mathbbm{r} \cdot \sigma \right) = 0$ and $\det( \mathbbm{r} \cdot \sigma ) = -\mathbbm{r}_y^2 - \mathbbm{r}_x^2 - \mathbbm{r}_z^2 = - \norm{\mathbbm{r}}_2^2$. This shows that the eigenvalues of $\mathbbm{r} \cdot \sigma$ are given by $\lambda_+ = \norm{\mathbbm{r}}_2$ and $\lambda_- = -\norm{\mathbbm{r}}_2$. Together with~\eqref{eqn::qubit_in_pauli_basis} this shows the first claim. Let us now consider
    \begin{align*}
        \norm{\rho(\mathbbm{r}_1) - \rho(\mathbbm{r}_2)}_{Tr} = \frac{1}{2} \norm{(\mathbbm{r}_1 - \mathbbm{r}_2)\cdot \sigma}_{Tr} = \frac{1}{4} (\norm{\mathbbm{r}_1 - \mathbbm{r}_2}_2 + \norm{\mathbbm{r}_1 - \mathbbm{r}_2}_2) = \frac{1}{2} \norm{\mathbbm{r}_1 - \mathbbm{r}_2}_2
    \end{align*}
    where we used the representation of $\rho(\mathbbm{r}_i)$ in equation~\eqref{eqn::qubit_in_pauli_basis}.
\end{proof}
Let us now state an intermediary results that we make use of throughout the paper. It relates the distance of the probability distributions induced by a measurement on two quantum states to the Trace-norm distance between the states. For readability, we defer the proof to Section~\ref{sec::additonal_proofs}.
\begin{lemma}
\label{lem::Help1}
    Let $M$ be a quantum measurement with measurement operators $M_y$, where $E_y := M_y^* M_y$. Furthermore, for $i \in \{1,2\}$, let $\rho_i \in \mathcal{S}(\mathbb{C}^d)$ be quantum states. We set $$
    p_i^{R_M}(y) := \mathcal{P}_i(R_M = y) = \Tr(\rho_i E_y).
    $$ 
    Then:
    \begin{enumerate}
        \item[(i)] For all $y$ it holds: $\lambda_{min}(E_y) \leq p_i^{R_M}(y) \leq \lambda_{max}(E_y)$,

        \item[(ii)] For all $y$ it holds: $\left| p_{1}^{R_M}(y) - p_{2}^{R_M}(y) \right| \leq \left( \lambda_{max}(E_y) - \lambda_{min}(E_y) \right) \norm{\rho_1 - \rho_2}_{Tr}$.
    \end{enumerate}
\end{lemma}

\subsection{Distance measures of probability distributions}
Let us shortly recall some classical properties of some distances between probability distributions. The most important one for us is the Kullback-Leibler divergence.
\begin{definition}
    Let $\mathcal{P}_{1}, \mathcal{P}_2$ be two mutually absolutely continuous probability distributions with corresponding probability mass functions (pmf) $p_1, p_2$. The Kullback-Leibler divergence between $\mathcal{P}_1$ and $\mathcal{P}_2$ is given by 
    \begin{equation*}
        D_{KL}(\mathcal{P}_1 ||\mathcal{P}_2) = \sum_{y} p_1(y) \log\left(\frac{p_2(y)}{p_1(y)}\right).
    \end{equation*}
    Since the Kullback-Leibler divergence is not symmetric, we often consider the symmetrized Kullback-Leibler divergence $D_{KL}^{sym}(\mathcal{P}_1 ||\mathcal{P}_2) = D_{KL}(\mathcal{P}_1 ||\mathcal{P}_2) + D_{KL}(\mathcal{P}_2 ||\mathcal{P}_1)$. 
\end{definition}

The smaller the Kullback-Leibler divergence of two probability distributions, the closer the two distributions are. As such, it allows us to asses the information in the outcome distribution of a given measurement. The better a measurement $M$ can distinguish two quantum states $\rho_1, \rho_2$ the further apart their outcome distributions $\mathcal{P}_{\rho_1}^{R_M}, \mathcal{P}_{\rho_2}^{R_M}$ are. The Kullback-Leibler divergence has the nice property of factorizing for product distributions.
\begin{lemma}{\normalfont (\cite{Tsybakov})}
    Let $\mathcal{P}_{1}^{\otimes n}, \mathcal{P}_2^{\otimes n}$ be two $n$-fold product distributions. Then it holds
    \begin{equation*}
        D_{KL}(\mathcal{P}_1^{\otimes n} ||\mathcal{P}_2^{\otimes n}) = n D_{KL}(\mathcal{P}_1 ||\mathcal{P}_2) .
    \end{equation*}
\end{lemma}
Another important classical distance measure is the total-variation distance.
\begin{definition}
    Let $\mathcal{P}_{1}, \mathcal{P}_2$ be two probability distributions with corresponding probability mass functions (pmf) $p_1, p_2$. The total-variation distance between $\mathcal{P}_1$ and $\mathcal{P}_2$ is given by 
    \begin{equation*}
        \norm{\mathcal{P}_1 - \mathcal{P}_2}_{TV} = \frac{1}{2} \sum_{y} |p_1(y) - p_2(y)|.
    \end{equation*}
\end{definition}
The total-variation distance is the classical analogue of the trace-norm distance and is closely linked to the optimal error in distinguishing between to probability distributions. Since the total-variation of product distributions is difficult to calculate, it is useful to relate the total-variation distance and Kullback-Leibler divergence
\begin{lemma}{\normalfont (Pinsker's inequality \cite{Tsybakov})} Let $\mathcal{P}_{1}, \mathcal{P}_2$ be two probability distributions. Then it holds
\begin{equation*}
    \norm{\mathcal{P}_1 - \mathcal{P}_2}_{TV} \leq \sqrt{\frac{1}{2} D_{KL}(\mathcal{P}_1 ||\mathcal{P}_2)}.
\end{equation*}
\end{lemma}

\section{Gentleness and quantum differential privacy}
\label{sec::gentleness}
Let us now introduce the concept of gentle measurements and see how they allow further processing when compared to destructive measurements. To do this, let us first demonstrate the impact of a general non-gentle measurement on a quantum state. Let $\ket{0}, \, \ket{1}$ be an arbitrary orthonormal basis of $\mathbb{C}^2$. We can then represent a pure state $\rho_{\eta} = \ket{\psi_{\eta}}\bra{\psi_{\eta}}$ for $\eta$ a complex number with $|\eta| \leq 1$ by the vector
\begin{equation*}
    \ket{\psi_{\eta}} = \eta \ket{0} + \sqrt{1 - |\eta|^2} \ket{1}.
\end{equation*}
Let us now see what happens when we measure the state in the computational basis i.e. using the measurement $M = (M_0, M_1)$, where $M_0 = \ket{0}\bra{0}$ and $M_1 = \ket{1}\bra{1}$. Applying the formulae from Definition~\ref{defn::quantum_measurement} shows that the two possible outcome states are
\begin{equation*}
    ({\rho_{\eta}})_{M \to 1} = \ket{1}\bra{1} \hspace{20pt} \text{and} \hspace{20pt} ({\rho_{\eta}})_{M \to 0} = \ket{0}\bra{0} 
\end{equation*}
with
\begin{equation*}
    \norm{\rho_{\eta} - ({\rho_{\eta}})_{M \to 1}}_{Tr}  = \eta \hspace{20pt} \text{and} \hspace{20pt} \norm{\rho_{\eta} - ({\rho_{\eta}})_{M \to 0}}_{Tr}  = \sqrt{1 - \eta^2}.
\end{equation*}
As such, for $|\eta|$ close to $0$ or $1$, one of the possible post-measurement states is almost maximally far away from the initial state. Since we do not know the state beforehand, we cannot exclude these edge cases from our model. Even more, the post-measurement state is completely unrelated to the state before the measurements, as the states after the measurements are independent of $\eta$. \footnote{\cite{aaronson2019gentle} speaks of "garbage" that is produced when measuring in such a way.} When we perform subsequent measurements of $M$, the outcome will always be $0$ (if the output state was $\ket{0}\bra{0}$) or always $1$ (if the output state was $\ket{1}\bra{1}$). It does not carry any useful information on the initial state $\rho_{\eta}$ anymore. Conversely, let us now assume that the measurement we perform is in such a way that, instead of completely collapsing the state $\rho_{\eta}$, the post measurement states are given by
\begin{equation}
\label{eqn::example_close_post_measurement_state}
    (\rho_{\eta})_{M \to y} = \rho_{\eta} + \Delta_{y, \alpha},
\end{equation}
where $\Delta_{y, \alpha} \in \mathbb{C}^{2 \times 2}$ is such that $\norm{\Delta_{y, \alpha}}_{Tr} \leq \alpha$. Let us see what now happens when we perform any subsequent measurement $\Tilde{M}$ on $(\rho_{\eta})_{M \to y}$. The outcome probabilities of the measurement $\tilde{M}$ will then be
\begin{align*}
    \mathcal{P}_{(\rho_{\eta})_{M \to y}}\left( R_{\tilde{M}} = \tilde{y} \right) &= \Tr\left( (\rho_{\eta})_{M \to y} \tilde{M}_{\tilde{y}}^* \tilde{M}_{\tilde{y}} \right)
    \\
    &= \Tr\left( \rho_{\eta} \tilde{M}_{\tilde{y}}^* \tilde{M}_{\tilde{y}} \right) + \Tr\left( \Delta_{y, \alpha} \tilde{M}_{\tilde{y}}^* \tilde{M}_{\tilde{y}} \right) 
    \\
    &= \mathcal{P}_{\rho_{\eta}}\left( R_{\tilde{M}} = \tilde{y} \right) + \Tr\left( \Delta_{y, \alpha} \tilde{M}_{\tilde{y}}^* \tilde{M}_{\tilde{y}} \right).
\end{align*}
Now, since we assumed $\norm{\Delta_{y, \alpha}}_{Tr} \leq \alpha$, we have
\begin{equation*}
    \left| \Tr\left( \Delta_{y, \alpha} \tilde{M}_{\tilde{y}}^* \tilde{M}_{\tilde{y}} \right) \right| \leq \alpha,
\end{equation*}
which shows that 
\begin{equation*}
    \mathcal{P}_{(\rho_{\eta})_{M \to y}}\left( R_{\tilde{M}} = \tilde{y} \right) = \mathcal{P}_{\rho_{\eta}}\left( R_{\tilde{M}} = \tilde{y} \right) + O(\alpha).
\end{equation*}
We see that under the assumption~\eqref{eqn::example_close_post_measurement_state} on the first measurement, the outcomes of the second measurement still carry information about the initial state. As such, we can gain information about the quantum state $\rho_{\eta}$ while retaining some information about it in the post-measurements state $(\rho_{\eta})_{M \to y}$ that allows for further quantum post-processing. Let us now formalize the above observation. 

\begin{definition}
\label{defn::gentleness}
    For a given \textit{gentleness parameter} $\alpha \in [0,1]$, a measurement $M$ is $\alpha$-gentle on a set $\mathcal{S}$ of quantum states if for all possible measurement outcomes $y$
    \begin{equation*}
        \norm{\rho - \rho_{M \to y}}_{Tr} \leq \alpha \hspace{20pt} \text{for all } \rho \in \mathcal{S}.
    \end{equation*}
    In that case, we write $M \in GM(\alpha, \mathcal{S})$. 
    
    If $\rho = \rho_1 \otimes ... \otimes \rho_n $ is a product state belonging to $ \mathcal{S}_1 \otimes ... \otimes \mathcal{S}_n =: \mathcal{S}^n$, we say that a measurement $M$ is locally-$\alpha$-gentle if it is a product measurement $M = M_1 \otimes ... \otimes M_n$ and $M_i$ is $\alpha$-gentle on $\mathcal{S}_i$ for all $i$. In that case, we write $M \in LGM(\alpha, \mathcal{S}^n)$.

    If $M$ is an arbitrary measurement of the product state $\rho$, that is $\alpha$-gentle, we write $M \in GM(\alpha, \mathcal{S}^n)$.
\end{definition}
By Busch's Theorem (see. \cite{Myrvold_2009}) we know that it is impossible to obtain information about a quantum system without disturbing it.\footnote{Busch's Theorem is a "No Free Lunch Theorem" in the quantum world. In a sense, we can think of this work as establishing the price of the lunch.} As such, the two extreme cases, $\alpha = 0$ and $\alpha = 1$, are covered by already known measurement techniques. That is, the case $\alpha = 0$ corresponds to the identity measurement that measures nothing. Here, we cannot infer anything on the system. The case $\alpha = 1$ corresponds to the case of non-gentle measurements without any restrictions. Here the sample optimality of measurement techniques is known \cite{Schmied_2016}. The goal of this paper is to bridge the gap between these two cases.

Considering globally gentle measurements in $GM(\alpha, \mathcal{S}^n)$ on the whole system (as in \cite{aaronson2019gentle}) requires access to (and manipulation of) a large system of multiple quantum states at the same time which is currently physically unfeasible or impossible in practice. In contrast, we only consider here local gentleness which better encapsulates the physical reality of quantum systems and their measurements since manipulations on individual subsystems are already physically possible.

\subsection{Properties of gentle measurements}
\label{sec::properties}
Let us develop some fundamental properties of gentle measurements. These properties will be crucial in proving lower bounds on the amount of copies needed for successful state certification and tomography. They are also instrumental for a better understanding of the differences between gentle and non-gentle measurements. We first note that when considering gentle measurements, we can restrict ourselves to gentleness on pure states. The proof can again be found in Section~\ref{sec::additonal_proofs}. 

\begin{lemma}
\label{lemma::3}
    If $M$ is an $\alpha$-gentle measurement on the set of pure states $\mathcal{S}_{pure}(\mathbb{C}^d)$ such that the $M_y$ are positive semi-definite, then $M$ is already $\alpha$-gentle on the whole set of states $\mathcal{S}(\mathbb{C}^d)$.
\end{lemma}

We see that the pure states, which constitute the extremal points of the convex set of states, are also being maximally altered by gentle measurements. Since pure states are often easier to work with than mixed states, studying the gentleness of a given measurement is significantly simplified by calculating only the gentleness on pure states. \medskip

We now consider the relation between gentle-measurements and quantum differentially-private measurements.

\begin{definition}
\label{defn::quantum_differential_privacy}
    A measurement $M $ is said to be $\delta$-quantum differentially-private ($\delta$-qDP) for $\delta>0$ on a set $\mathcal{S}$ of quantum states, if for the outcome probabilities of any two states $\rho_1, \rho_2 $ in $\mathcal{S}$ under $M$ it holds
    \begin{equation}
    \label{eqn::quantum_privacy_equation}
        \mathcal{P}_{1}(R_{M} = y) \leq e^{\delta} \mathcal{P}_2(R_{M} = y) \hspace{20pt} \text{for all } y \in \mathcal{Y}.
    \end{equation}
    As with gentleness, we call a measurement $M = M_1 \otimes ... \otimes M_n$ locally $\delta$-differentially private on $ \mathcal{S}_1 \otimes ... \otimes \mathcal{S}_n =: \mathcal{S}^n$ if each $M_i$ is differentially private on $\mathcal{S}_i$.
\end{definition}

\begin{remark}
\label{rem::privacy_as_eigenvalue_ratio}
    Since $\mathcal{P}_{\ket{\psi}}(R_M = y) = \bra{\psi}M_y^*M_y\ket{\psi}$ and $\lambda_{min}(M_y^*M_y) \leq \bra{\psi}M_y^*M_y\ket{\psi} \leq \lambda_{max}(M_y^*M_y)$ the probabilities $\mathcal{P}_i(R_{M} = y)$ are bounded from above and below by the maximal and minimal eigenvalues of $M_y^*M_y$ respectively. As such, if $\mathcal{S} = \mathcal{S}(\mathbb{C}^d)$, we can show quantum differential privacy by bounding the ratio of eigenvalues of $M_y^*M_y$. As such, for $\delta$ quantum-differentially private measurements $M = (M_y)_{y \in \mathcal{Y}}$ on $\mathcal{S}(\mathbb{C}^d)$, we have
    \begin{equation*}
        \frac{\lambda_{max}(M_y^*M_y)}{\lambda_{min}(M_y^*M_y)} \leq e^{\delta}.
    \end{equation*}
\end{remark}

Since local gentleness and local quantum Differential Privacy are defined separately for every register, it suffices to check \eqref{eqn::quantum_privacy_equation} for each register independently to show the latter.

\begin{lemma}
\label{lem::Gentle_eigenvalue_ratio}
    Let $\alpha $ in $[0, \frac{1}{2})$ and $M$ be $\alpha$-gentle on $\mathcal{S}(\mathbb{C}^d)$ with measurement operators $M_y$. Then $M$ is $\delta$ quantum differentially-private on $\mathcal{S}(\mathbb{C}^d)$ for $\delta = 2\log(\frac{1+2\alpha}{1-2\alpha})$. 
\end{lemma}

\begin{proof}
    Assume that $M = (M_y)_{y \in \mathcal{Y}}$ is an $\alpha$-gentle measurement on $\mathcal{S}(\mathbb{C}^d)$. Define $E_y := M_y^*M_y$. Let $\rho_1, \rho_2 \in \mathcal{S}(\mathbb{C}^d)$ s.t. $\norm{\rho_1 - \rho_2}_{Tr} = 1$. Let us further denote by $p_1$ (respectively $p_2$) the probability of obtaining outcome $y$ under $\rho_1$ (respectively $\rho_2$). Without loss of generality we assume that $p_1 > p_2 \geq 0$. Now, let
    \begin{equation*}
        \rho_{\lambda} = \lambda \rho_1 + (1-\lambda) \rho_2 \hspace{10pt} \text{for all } \lambda \in (0,1).
    \end{equation*}
    The probability of obtaining the outcome $y$ when measuring $\rho_{\lambda}$ is
    \begin{equation*}
        p_{\lambda} = \mathcal{P}_{\lambda}(R_M = y) = \Tr\left( \rho_{\lambda} E_y \right) = \lambda \Tr\left( \rho_1 E_y \right) + (1 - \lambda) \Tr\left( \rho_2 E_y \right) = \lambda p_1 + (1 - \lambda) p_2.
    \end{equation*}
    The post-measurement state of $\rho_{\lambda}$ is then given by
    \begin{equation*}
        (\rho_{\lambda})_{M \to y} = \frac{1}{p_{\lambda}} M_y \rho_{\lambda} M_y^* = \frac{\lambda p_1 (\rho_1)_{M \to y} + (1-\lambda) p_2 (\rho_2)_{M \to y}}{\lambda p_1 + (1-\lambda) p_2}.
    \end{equation*}
    Now if we define $\delta = \frac{\lambda p_1}{\lambda p_1 + (1-\lambda) p_2} - \lambda > 0$, we get
    \begin{align*}
        \rho_{\lambda} - (\rho_{\lambda})_{M \to y} = \frac{\lambda p_1}{p_{\lambda}} \left((\rho_1 - (\rho_1)_{M \to y}\right) + \frac{(1-\lambda) p_2}{p_{\lambda}} \left(\rho_2 - (\rho_2)_{M \to y}\right) + \delta \left(\rho_2 - \rho_1\right)).
    \end{align*}
    By the triangle inequality and gentleness we now have
    \begin{align*}
        \delta \norm{\rho_2 - \rho_1}_{Tr} \leq \frac{\lambda p_1}{p_{\lambda}} \alpha + \frac{(1-\lambda) p_2}{p_{\lambda}} \alpha + \alpha = 2 \alpha
    \end{align*}
    Since we further assumed $\norm{\rho_2 - \rho_1}_{Tr} = 1$, we get $\delta \leq 2 \alpha$. This allows us to write
    \begin{align*}
        p_1 = \frac{\lambda - \lambda^2 + \delta (1-\lambda)}{\lambda - \lambda^2 - \delta \lambda} p_2 \leq \frac{\lambda - \lambda^2 + 2\alpha (1-\lambda)}{\lambda - \lambda^2 - 2 \alpha \lambda} p_2 \hspace{10pt} \text{for all } \lambda \in (0,1 - 2 \alpha).
    \end{align*}
    The last inequality only holds as long as the denominator is positive which is the case for $\lambda < 1 - 2 \alpha$. As such, for $\lambda_0 = \frac{1 - 2 \alpha}{2} < 1 - 2 \alpha$, we obtain
    \begin{equation*}
        p_1 \leq \frac{\lambda_0 - \lambda_0^2 + 2\alpha (1-\lambda_0)}{\lambda_0 - \lambda_0^2 - 2 \alpha \lambda_0} p_2 = \left( \frac{1 + 2\alpha}{1 - 2\alpha} \right)^2 p_2.
    \end{equation*}
    Now, since we started with $\rho_1, \rho_2$ such that $\norm{\rho_1 - \rho_2}_{Tr} = 1$, the last relation holds for every pure state. As such $M$ is $2 \log\left( \frac{1 + 2\alpha}{1 - 2\alpha} \right)$-quantum differentially-private on pure states. Lemma~\ref{lem::Help1} then proves that $M$ is $2 \log\left( \frac{1 + 2\alpha}{1 - 2\alpha} \right)$-quantum differential private on the whole space.
\end{proof}

Lemma \ref{lem::Gentle_eigenvalue_ratio} is an improvement over the relation shown in \cite{aaronson2019gentle}. It not only further sharpens the constant $\delta$ but also extends the applicability of the result from $\alpha $ in $[0,\frac{1}{4})$ to $\alpha$ in $[0, \frac{1}{2})$. 

\begin{corollary}
    \label{lemma::1}
    Let $M$ be $\alpha$-gentle on $\mathcal{S}(\mathbb{C}^d)$ with measurement operators $M_y$ for $\alpha $ in $[0, 1)$. Then $M_y$ has full rank.
\end{corollary}

This is a direct consequence of Remark~\ref{rem::privacy_as_eigenvalue_ratio} as the rank of $M_y$ is the same as the rank of $E_y = M_y^*M_y$ whose eigenvalues are all non-zero. Corollary \ref{lemma::1} already severely reduces the amount of possible gentle measurements. It also shows that many of the well known measurements, such as PVMs and mutually unbiased basis measurements, cannot be gentle. This means that we have to develop new measurement techniques and cannot directly rely on existing results. 

If we assume that the measurement operators $M_y$ are itself hermitian and positive-definite (which we call positive for short) the result of Lemma \ref{lem::Gentle_eigenvalue_ratio} can be further improved as follows.

\begin{lemma}
\label{lem::improved_constant_positivity}
    Let $\alpha $ in $[0, 1)$ and $M$ be $\alpha$-gentle on $\mathcal{S}(\mathbb{C}^d)$ with positive-definite measurement operators $M_y$. Then $M$ is $\delta$ quantum differentially-private on $\mathcal{S}(\mathbb{C}^d)$ for $\delta = 2\log(\frac{1+\alpha}{1-\alpha})$.
\end{lemma}
\begin{proof}
    Let $M_y = \sum_{i = 1}^d \lambda_i \ket{v_i}\bra{v_i}$ be positive-definite with maximal and minimal eigenvalue $\lambda_1$ and $\lambda_d$ respectively. Consider the gentleness of $M_y$ on the pure state $\ket{\psi} = \frac{1}{\sqrt{\lambda_1 + \lambda_d}} (\sqrt{\lambda_d} \ket{v_1} + \sqrt{\lambda_1} \ket{v_d})$. Using the Lemma \ref{lemmaTraceNorm} then gives 
    \begin{equation*}
        \frac{\lambda_1}{\lambda_d} \leq \frac{1 + \alpha}{1- \alpha}.
    \end{equation*}
    The result then follows using Remark~\ref{rem::privacy_as_eigenvalue_ratio}.
\end{proof}
Since $M_y^*M_y = |M_y|^2$, a measurement consisting of positive operators $M_y$ has the same output probabilities as the version where the measurement operators are not assumed to be positive. We conjecture that the measurement consisting of the operators $|M_y|$ has at least the same gentleness as the measurement given by the operators $M_y$. Under the assumptions of Lemma \ref{lem::improved_constant_positivity} the constant $\delta$ is asymptotically sharp as $\alpha \to 0$. The qLS mechanism we define attains this value. 

\subsection{Gentle-izing a two outcome measurement}
%\textcolor{red}{new subsection maybe? (Gentle-izing a two outcome measurement)}
%\textcolor{red}{Introduction why we need this}
We will now demonstrate shortly how we can create a gentle measurement by modifying known non-gentle measurements. Given the relation between gentleness and differential privacy, we may use known mechanisms from local differential privacy to define gentle measurements. Since the resulting random variable of any two-outcome measurements is Bernoulli distributed, we may define measurements based on the label switching mechanism for Bernoulli random variables (that is shown to be optimal \cite{steinberger2023efficiency}) in order to develop gentle measurements. The technique works for two outcome measurements in any dimension and we will make use of the following lemma several times throughout this paper.
\begin{lemma}
\label{lem::qLS_gentleness}
Let $\alpha > 0$ and $M = (P_1, P_0)$ be a two-outcome PVM. For $\delta = 4 \arctanh(\alpha)$ we define the measurement $M_{\alpha} = (M_{\alpha,1}, M_{\alpha, 0})$ by
\begin{align*}
    M_{\alpha, 1} = \sqrt{\frac{e^{\delta}}{e^{\delta}+1}}  P_1 + \sqrt{\frac{1}{e^{\delta}+1}} P_0 \hspace{20pt} \text{and} \hspace{20pt} M_{\alpha, 0} = \sqrt{\frac{1}{e^{\delta}+1}}  P_1 + \sqrt{\frac{e^{\delta}}{e^{\delta}+1}} P_0.
\end{align*}
Then $M_{\alpha}$ is $\alpha$ gentle on $\mathcal{S}(\mathbb{C}^d)$ for all $\alpha \in [0, 1]$, that is $M_{\alpha} $ belongs to $ GM(\alpha, \mathcal{S}(\mathbb{C}^d))$.
\end{lemma}
\begin{proof}
    Since $P_1$ and $P_0$ are orthogonal projectors that fullfil the resolution of the identity we can find an orthonormal basis $(\ket{k})_{k = 1,...,d}$ with
    \begin{equation*}
        P_1 = \sum_{k = 1}^{d_+} \ket{k}\bra{k} \hspace{20pt} \text{and} \hspace{20pt} P_0 = \sum_{k = d_+ + 1}^{d} \ket{k}\bra{k}.
    \end{equation*}
    We will now show the gentleness of $M_{\alpha}$ on pure states and use Lemma \ref{lemma::3} to extend the result onto the whole space of states $\mathcal{S}(\mathbb{C}^d)$. We can write any pure state $\ket{\psi}$ in its basis representation with respect to the $\ket{k}$ as
    \begin{equation*}
        \ket{\psi} = \sum_{k = 1}^d \gamma_k(\psi) \ket{k}, \hspace{10pt} \text{where } \hspace{10pt} \sum_{k = 1}^d |\gamma_k(\psi)|^2 = 1.
    \end{equation*}
    The outcome probabilities of the measurement $M_{\alpha}$ are then given by
    \begin{align*}
        \mathcal{P}_{\ket{\psi}}(R^{M_{\alpha}} = 1) &= \bra{\psi} M_{\alpha, 1}^2\ket{\psi} = \frac{1}{e^{\delta}+1} \left( e^{\delta} \sum_{k = 1}^{d_+} |\gamma_k(\psi)|^2 + \sum_{k = d_+ + 1}^d |\gamma_k(\psi)|^2 \right)
        \\
        \mathcal{P}_{\ket{\psi}}(R^{M_{\alpha}} = 0) &= \bra{\psi} M_{\alpha, 0}^2\ket{\psi} = \frac{1}{e^{\delta}+1} \left( \sum_{k = 1}^{d_+} |\gamma_k(\psi)|^2 + e^{\delta} \sum_{k = d_+ + 1}^d |\gamma_k(\psi)|^2 \right)
    \end{align*}
    and the post-measurement state is given by
    \begin{align*}
        \ket{\psi}_{M_{\alpha} \to 1} &= \frac{1}{\sqrt{\mathcal{P}_{\ket{\psi}}(R^{M_{\alpha}} = 1)}} \left( \sum_{k = 1}^{d_+} \sqrt{\frac{e^{\delta}}{e^{\delta}+1}} \gamma_k(\psi) \ket{k} + \sum_{k = d_+ + 1}^{d} \sqrt{\frac{1}{e^{\delta}+1}} \gamma_k(\psi) \ket{k} \right)
        \\
        \ket{\psi}_{M_{\alpha} \to 0} &= \frac{1}{\sqrt{\mathcal{P}_{\ket{\psi}}(R^{M_{\alpha}} = 1)}} \left( \sum_{k = 1}^{d_+} \sqrt{\frac{1}{e^{\delta}+1}} \gamma_k(\psi) \ket{k} + \sum_{k = d_+ + 1}^{d} \sqrt{\frac{e^{\delta}}{e^{\delta}+1}} \gamma_k(\psi) \ket{k} \right).
    \end{align*}
    Let us calculate the distance between the pre- and post-measurement states for the outcome $y=1$. The outcome $0$ follows completely analogously. Since both the pre- and post-measurement state are pure, we use Lemma \ref{lemmaTraceNorm} to calculate
    \begin{align*}
        \norm{\ket{\psi}\bra{\psi} - \ket{\psi}\bra{\psi}_{M_{\alpha} \to 1}}_{Tr}^2 &= 1 - \left| \bra{\psi}\ket{\psi_{M_{\alpha} \to 1}} \right|^2
        \\
        &= 1 - \frac{\left(e^{\frac{\delta}{2}} \sum_{k = 1}^{d_+} |\gamma_k(\psi)|^2 + \sum_{k = d_+ + 1}^d |\gamma_k(\psi)|^2 \right)^2}{e^{\delta} \sum_{k = 1}^{d_+} |\gamma_k(\psi)|^2 + \sum_{k = d_+ + 1}^d |\gamma_k(\psi)|^2}
        \\
        &= 1 - \frac{(e^{\frac{\delta}{2}} \gamma_+ + 1 - \gamma_+)^2}{e^{\delta} \gamma_+ + 1 - \gamma_+} =: g_{\delta}(\gamma_+)
    \end{align*}
    where we used the notation
    \begin{equation*}
        \gamma_+ = \gamma_+(\psi) = \sum_{k = 1}^{d_+} |\gamma_k(\psi)|^2.
    \end{equation*}
    We see that the gentleness of the measurement is only dependent on the value $\gamma_+ \in [0,1]$. When we maximize the function $g_{\delta}$ on $[0,1]$, we see that the maximum is attained at
    \begin{equation*}
        \gamma_{max} = \frac{1}{e^{\frac{\delta}{2}}+ 1}
    \end{equation*}
    and we have 
    \begin{equation*}
        \norm{\ket{\psi}\bra{\psi} - \ket{\psi}\bra{\psi}_{M_{\alpha} \to 1}}_{Tr}^2 \leq g_{\delta}(\gamma_{max}) = \left(\frac{e^{\frac{\delta}{2}} - 1}{e^{\frac{\delta}{2}} +  1} \right)^2 = \tanh^2\left( \frac{\delta}{4} \right) = \alpha^2.
    \end{equation*}
    Since the result holds analogously for outcome $0$, the $\alpha$-gentleness of $M_{\alpha}$ is shown.
\end{proof}

\subsection{Local and global gentleness and Differential privacy}
\label{sec::Privacy_Triviality}
In this part, we discuss in more detail the relation between differential privacy and gentleness. In particular, we will take a look at the differences between local and global differentially private mechanisms, i.e. Markov Kernels, as well as local and global gentle measurements. Since the measurements we consider are non-interactive, we restrict our comparison to non-interactive privacy mechanisms as well. Furthermore, we assumed countable outcomes of the measurement in order to define the post-measurement state. In order to highlight the connection to privacy, we assume that the outcome of the privacy mechanism is countable as well.

\begin{definition}
    Let $\delta > 0$. A mechanism, i.e. a Markov kernel, $\mathcal{A}$ on $\mathcal{X}^n$ with values in $\mathcal{Y}$ is said to be (globally) $\delta$ differentially private if for any two $x, x' \in \mathcal{X}^n$ such that $x$ and $x'$ only differ in at most one entry it holds
    \begin{equation*}
        \mathcal{P}(\mathcal{A}(x) = y) \leq e^{\delta} \mathcal{P}(\mathcal{A}(x') = y)
    \end{equation*}
    for any outcome $y \in  \mathcal{Y}$.
\end{definition}
The elements $x, x'$ that differ in at most only one entry, i.e. $\norm{d - d'}_0 \leq 1$, are said to be \textit{adjacent}. The set $\mathcal{X}^n$ can be thought of as a database. This definition assures that we cannot draw conclusions about the entries of the database $\mathcal{D}$ when given only the privatized data, while still being able to perform statistical analyses on the data. This mechanism is used to protect confidentiality for medical data \cite{DYDA2021100366} as well as by tech-companies to analyze their products without invading the users privacy \cite{Apple_website}. The downside of global privacy is, that the user still has to give the un-privatized data to someone else and trust this central entity privatize the data before analyzing it. Local Differential Privacy remedies this problem by privatizing the data before it is transmitted into the database.
\begin{definition}
\label{def::local_differential_privacy}
    Let $\delta > 0$, a product mechanism $\mathcal{A} = \mathcal{A}_1 \times ... \times \mathcal{A}_n$ on the space $\mathcal{X}_1 \times ... \times \mathcal{X}_n = \mathcal{X}^n$ with outcomes in $\mathcal{Y}^n = \mathcal{Y}_1 \times ... \times \mathcal{Y}_n$ is said to be locally-$\delta$ differentially-private, if for any $i$ and for any two inputs $x, x' \in \mathcal{X}_i$ it holds
    \begin{equation*}
        \mathcal{P}(\mathcal{A}_i(x) = y_i) \leq e^{\delta} \mathcal{P}(\mathcal{A}_i(x') = y_i)
    \end{equation*}
    for any outcome $y_i \in  \mathcal{Y}_i$.
\end{definition}
\begin{remark}
    Note that local privacy is usually defined for a single mechanism $\mathcal{A}_i$ only. This leads to the same object as Definition \ref{def::local_differential_privacy} by considering product mechanism but allows for a simpler adaption to interactive mechanisms. However, we chose this formulation to better highlight the connection to gentleness.
\end{remark}
Local differential privacy is stricter than global differential privacy. It assures a higher level of privacy for the user at the cost of information about the true data. When we think of a locally-DP mechanism we think of having $n$ individuals $x_i$ and for each one, we apply the local mechanism. %We could apply a local privacy mechanism to a whole database and view it as one ``local'', multivariate system $\mathcal{X} = \mathcal{X}^n$. 
When we apply a privacy mechanism to a whole database viewed as one multivariate system $\mathcal{X} = \mathcal{X}^n$ we call the privacy mechanism global or central. This is the classical equivalent of our locally gentle measurements acting independently on each state and globally gentle measurements acting on the whole system.
%This is the classical equivalent of what we do, when we show that consistent global gentle state certification and tomography is not possible (see Remark~\ref{rem::DPI_applied_to_coherent_measurement}). 

When defining quantum local differential privacy, we can do the same distinction as in the classical case, with some changes. The privacy mechanism (or Markov Kernel) is replaced by a quantum measurement whereas the input samples are replaced by input quantum states.  Let us restate in more details the Definition~\ref{defn::quantum_differential_privacy}. 

\begin{definition}  
\label{defn::quantum_differential_privacy_in_depth}
    A product measurement $M = M_1 \otimes ... \otimes M_n$ on $\mathcal{S}(\mathbb{C}^d) \otimes ... \otimes \mathcal{S}(\mathbb{C}^d)$ with outcomes in $ \mathcal{Y}_1 \times ... \times \mathcal{Y}_n $ is said to be locally-$\delta$ quantum differentially-private ($\delta$-qDP) for $\delta>0$, if for the outcome probabilities of any two states $\rho, \rho' $ in $\mathcal{S}(\mathbb{C}^d)$ under $M_i$ it holds
    \begin{equation*}
        \mathcal{P}_{\rho}(R_{M_i} = y_i) \leq e^{\delta} \mathcal{P}_{\rho'}(R_{M_i} = y_i) \hspace{20pt} \text{for all } y_i \in \mathcal{Y}_i.
    \end{equation*}
\end{definition}
This is our definition of {\it local} quantum differential privacy. In contrast, \cite{aaronson2019gentle} consider measurement on products of states and we call their concept {\it global} differential privacy. 

\begin{definition}{\normalfont ($\delta$ quantum differential privacy as in \cite{aaronson2019gentle})}
    Let $\delta > 0$. A quantum measurement on a set of product states $\mathcal{S}^n$ with outcomes in $\mathcal{Y}$ is said to be globally-$\delta$ quantum differentially-private if for any $\rho, \rho' \in \mathcal{S}^{n}$ with $\rho = \rho_{1} \otimes ... \otimes \rho_{n}$ and $ \rho' = \rho_{1}' \otimes ... \otimes \rho_{n}'$ such that $\rho_1$ and $\rho_2$ differ in at most only one register it holds
    \begin{equation*}
        \mathcal{P}_{\rho}(R_M = y) \leq e^{\delta} \mathcal{P}_{\rho'}(R_M = y)
    \end{equation*}
    for any outcome $y \in \mathcal{Y}$.
\end{definition}
As in the classical case, the $\rho = \rho_{1} \otimes ... \otimes \rho_{n}$ and $ \rho' = \rho_{1}' \otimes ... \otimes \rho_{n}'$ that differ in at most only one register are said to be \textit{adjacent} or \textit{neighbors}. Differing in only one register means that $\rho_{i} = \rho_{i}'$ for all but at most one $i$. 

When connecting quantum differential privacy to gentleness, it is important which definition (local or global) we take as a starting point. Theorem 5 in \cite{aaronson2019gentle} and Lemma \ref{lem::Gentle_eigenvalue_ratio} both establish a connection between quantum differential privacy and gentleness in their respective global and local situation. Let us summarize this here.
\begin{proposition}
    Let $\alpha < 1/2$ and $M$ be a quantum measurement on a $d$-dimensional quantum system.
    \begin{enumerate}
        \item[(i)] \textbf{[Theorem 5 in \cite{aaronson2019gentle}]} If $M$ is globally-$\alpha$-gentle on product states, then $M$ is globally-$2 \log\left( \frac{1+2\alpha}{1-2\alpha} \right)$-differentially-private on product states, that is $M \in GM(\alpha, \mathcal{S}^n)$, where $\mathcal{S}^n$ denotes the set of product states in $\mathcal{S}((\mathbb{C}^d)^{\otimes n})$.  

        \item[(ii)] \textbf{[Lemma \ref{lem::Gentle_eigenvalue_ratio}]} If $M = M_1 \otimes ... \otimes M_n$ is locally-$\alpha$-gentle, each $M_i$ is locally-$2 \log\left( \frac{1+2\alpha}{1-2\alpha} \right)$-differentially-private, that is $M \in LGM(\alpha, \mathcal{S}(\mathbb{C}^d))$. 
    \end{enumerate}
\end{proposition}
    
The subtle but important difference is that global gentleness alters the whole system by $\alpha$ while local gentleness alters each subsystem by $\alpha$. Lemma \ref{lemmaTraceNorm} shows that, in general, the alteration of each subsystem is amplified in the whole system. As such, globally-gentle measurements destroy the measured system less than their local counterparts. For them to be implemented physically we need to exploit quantum effects such as entanglement which require access to the whole system at once and perform coherent operations on the whole system which are currently physically unfeasible. In the same way that a  classical globally differentially private mechanism needs access to the whole database at once, a globally-gentle measurement needs access to the whole quantum system at once. The privacy-information trade-off in the classical world is therefore mirrored by a physical feasibility-information trade-off in the quantum world. 

\begin{figure}[ht]
    \setlength{\fboxsep}{9pt}%
    \setlength{\fboxrule}{1pt}%
    \fbox{\resizebox{.95\linewidth}{!}{\begin{tikzpicture}[
    state/.style={rectangle, draw, thick, minimum width=1.5cm, minimum height=1cm, font=\Large},
    measure/.style={rectangle, draw, thick, minimum width=1.5cm, minimum height=1.5cm},
    qline/.style={-, thick},
    cline/.style={double, double distance=1.5pt, thick, -{Latex[length=4mm, width=4mm]}}
]

\newcommand{\drawmeter}[1]{
    \draw[thick] ([xshift=-0.5cm, yshift=-0.4cm]#1.center) arc[start angle=150, end angle=30, radius=0.6cm];
    \draw[thick, -{Stealth[length=3mm]}] ([yshift=-0.5cm, xshift=0cm]#1.center) -- ([xshift=0.4cm, yshift=0.2cm]#1.center);
}

\node[font = \Large] (s1) at (0, 0) {$\rho_1$};
\node[font = \Large] (p1) at (1.5, 0) {$\otimes$};
\node[measure, label={[xshift=-0.3cm, yshift=-0.7cm]{\Large $M_1$}}] (m1) at (0, -3) {};
\node[font = \large] (r1) at (1.8, -3) {$y_1$};
\node[font = \large] (s1y) at (0, -6) {$(\rho_1)_{M_1 \to y_1}$};
\drawmeter{m1}

\node[font = \Large] (s2) at (3, 0) {$\rho_2$};
\node[font = \Large] (p1) at (4.5, 0) {$\otimes$};
\node[measure, label={[xshift=-0.3cm, yshift=-0.7cm]{\Large $M_2$}}] (m2) at (3, -3) {};
\node[font = \large] (r2) at (4.8, -3) {$y_2$};
\node[font = \large] (s2y) at (3, -6) {$(\rho_2)_{M_2 \to y_2}$};
\drawmeter{m2}

\node[font = \Large] (r3) at (5.5, 0) {$\dots$};
\node[font = \Large] (m3) at (5.5, -3) {$\dots$};
\node[font = \Large] (y3) at (5.5, -6) {$\dots$};

\node[font = \Large] (s3) at (7, 0) {$\rho_n$};
\node[font = \Large] (p1) at (6.3, 0) {$\otimes$};
\node[measure, label={[xshift=-0.3cm, yshift=-0.7cm]{\Large $M_n$}}] (m3) at (7, -3) {};
\node[font = \large] (r3) at (8.8, -3) {$y_n$};
\node[font = \large] (s3y) at (7, -6) {$(\rho_n)_{M_n \to y_n}$};
\drawmeter{m3}

\draw[qline] (s1) -- (m1);
\draw[cline] (m1) -- (r1);
\draw[qline] (m1) -- (s1y);

\draw[qline] (s2) -- (m2);
\draw[cline] (m2) -- (r2);
\draw[qline] (m2) -- (s2y);

\draw[qline] (s3) -- (m3);
\draw[cline] (m3) -- (r3);
\draw[qline] (m3) -- (s3y);

\node[font = \Large] (s1) at (10, 0) {$\rho_1$};
\node[font = \Large] (p1) at (11.5, 0) {$\otimes$};

\node[font = \Large] (s2) at (13, 0) {$\rho_2$};
\node[font = \Large] (p1) at (14.5, 0) {$\otimes$};

\node[font = \Large] (r3) at (15.5, 0) {$\dots$};

\node[font = \Large] (s3) at (17, 0) {$\rho_n$};
\node[font = \Large] (p1) at (16.3, 0) {$\otimes$};

\draw[line width = 0.03cm] (10, -0.3) -- (10, -1);
\draw[line width = 0.03cm] (13, -0.3) -- (13, -1);
\draw[line width = 0.03cm] (17, -0.3) -- (17, -1);
\draw[line width = 0.03cm] (10, -1) -- (17, -1);
\draw[line width = 0.03cm] (13.5, -1) -- (13.5, -1.6);

\node[font = \Large] (s1) at (13.5, -2) {$\rho$};

\draw[line width = 0.03cm] (13.5, -2.4) -- (13.5, -3);

\node[measure, label={[xshift=-0.3cm, yshift=-0.7cm]{\Large $M$}}] (m) at (13.5, -3.75) {};
\drawmeter{m}

\node[font = \large] (r) at (16, -3.75) {$y$};
\node[font = \large] (sy) at (13.5, -6) {$(\rho)_{M \to y}$};

\draw[cline] (m) -- (r);
\draw[qline] (m) -- (sy);

\end{tikzpicture}}}
    \label{alg::Quantum_Label_Switch}
    \caption{\justifying Schematic representation of local (left) and global (right) measurement schemes. Local schemes measure each subsystem $\rho_i$ individually whereas global schemes perform one measurement on the global system $\rho = \rho_1 \otimes ... \otimes \rho_n$. We say that the local measurement scheme is locally-gentle/quantum-differentially-private if each $M_i$ is gentle/quantum-differentially-private on each individual subsystem $\rho_i$. The global measurement scheme is globally-gentle/quantum-differentially private if it is gentle/quantum-differentially-private on the whole system $\rho$.}
\end{figure}

\section{A gentle data-processing inequality}
\label{sec::qDPI}
The main result of this work is a quantum data processing inequality for gentle measurements. The qDPI relates the symmetrized Kullback-Leibler divergence for the outcomes distributions of an $\alpha$-gentle measurement for two different states to the gentleness parameter $\alpha$ and the trace-norm distance between the two initial states. Since the Kullback-Leibler divergence plays a central role in finding optimal results in both quantum tomography and quantum state certification, this gives a new way of directly relating the information-gain of a measurement procedure to the amount of destruction of the system measured.

\begin{theorem}
\label{thm::B3_C2}
    Let $\mathcal{H} = \mathbb{C}^d$, $\mathcal{S}$ a set of quantum states that contains the pure states $\mathcal{S}_{pure}(\mathbb{C}^d)$ on $\mathbb{C}^d$. For $\alpha $ in $[0, \frac{1}{2})$ let $M $ belong to $ GM(\alpha, \mathcal{S})$. Then for any $\rho_1, \rho_2 $ in $\mathcal{S}$ it holds
    \begin{equation}
    \label{eqn::main_result}
        D_{KL}^{sym}(\mathcal{P}_{1}^{R_M} || \mathcal{P}_{2}^{R_M} ) \leq \left( \frac{8 \alpha}{(1-2\alpha)^2} \right)^2 \norm{\rho_1 - \rho_2}_{Tr}^2 \lesssim \alpha^2 \norm{\rho_1 - \rho_2}_{Tr}^2.
    \end{equation}
    The second inequality holds up to a numerical constant for small $\alpha$.
\end{theorem}

While classical privacy mechanisms act like contractions on the space of probability distributions, gentle measurements act as contractions when transporting quantum states into classical likelihoods. The factor of contraction was shown to be of order $\alpha^2$ for classical privacy mechanisms in \cite{duchi2014localprivacydataprocessing} and here we generalize this result to quantum gentle measurements.

\begin{proof}[Proof of Theorem~\ref{thm::B3_C2}]
    We obtain the data-processing inequality for gentle measurements by proving a data-processing inequality for quantum-differentially-private measurements and using the relation between the two. Let $M \in GM(\alpha, \mathcal{S})$. By Lemma~\ref{lem::Gentle_eigenvalue_ratio}, we know that $M$ is $\delta$ quantum differential-private for 
    \begin{equation*}
        \delta_{\alpha} = 2 \log\left( \frac{1 + 2\alpha}{1- 2\alpha} \right).
    \end{equation*}
    Note that for a $\delta$ quantum differentially-private measurement $M$, we have 
    
    \begin{equation*}
        0 < \lambda_{min}(E_y) \leq \mathcal{P}_i(R_m = y) = \Tr(\rho_i E_y) \leq \lambda_{max}(E_y),
    \end{equation*}
    where $E_y = M_y^*M_y$. Then we write for the symmetrized Kullback-Leibler-divergence
    \begin{align*}
        D_{KL}^{sym}(\mathcal{P}_1^{R_M} || \mathcal{P}_2^{R_M} ) =& \sum_y (p_1^{R_M}(y) - p_2^{R_M}(y)) \log\left( \frac{p_1^{R_M}(y)}{p_2^{R_M}(y)} \right)
        \\
        \leq& \sum_{y} |p_1^{R_M}(y) - p_2^{R_M}(y)|^2 \frac{1}{\min\{ p_1^{R_M}(y), p_2^{R_M}(y) \}}
        \intertext{where we used Lemma 4 in \cite{duchi2014localprivacydataprocessing}. We can now further bound this term using Lemma \ref{lem::Help1} as}
        D_{KL}^{sym}(\mathcal{P}_1^{R_M} || \mathcal{P}_2^{R_M} ) \leq& \sum_{y} \left( \lambda_{max}(E_y) - \lambda_{min}(E_y) \right)^2 \norm{\rho_1 - \rho_2}_{Tr}^2 \frac{1}{\lambda_{min}(E_y)}
        \\
        =& \sum_y \frac{\left( \lambda_{max}(E_y) - \lambda_{min}(E_y) \right)^2}{\lambda_{min}(E_y)^2} \norm{\rho_1 - \rho_2}_{Tr}^2 \lambda_{min}(E_y)
        \\
        =& \sum_y \left( \frac{\lambda_{max}(E_y)}{\lambda_{min}(E_y)} - 1 \right)^2 \norm{\rho_1 - \rho_2}_{Tr}^2 \lambda_{min}(E_y)
        \\
        \leq& \sum_y \left( e^{\delta_{\alpha}} -1 \right)^2 \norm{\rho_1 - \rho_2}_{Tr}^2 \lambda_{min}(E_y).
        \intertext{Finally we note that the $\sum_y\lambda_{min}(E_y)\leq \sum_y \Tr(E_y)/2 \leq \Tr( \text{Id})/2 = 1$ and obtain the final inequality}
        D_{KL}^{sym}(\mathcal{P}_1^{R_M} || \mathcal{P}_2^{R_M} ) \leq& \left( e^{\delta_{\alpha}} -1 \right)^2 \norm{\rho_1 - \rho_2}_{Tr}^2.
    \end{align*}
    Having shown the data-processing inequality for quantum differentially private measurements allows us now to easily show the analogous result for gentle measurements and the fact that
    \begin{equation*}
        \left( e^{\delta_{\alpha}} - 1 \right)^2 = \left( \frac{8\alpha}{(1 - 2\alpha)^2} \right)^2
    \end{equation*}
\end{proof}

Since the Kullback-Leibler divergence for measurements $M \in LGM(\alpha, \mathcal{S}^n)$ on product states $\rho_1^{\otimes n}$ and $\rho_2^{\otimes n}$ expands into the sum of the individual Kullback-Leibler divergences, we have
\begin{equation*}
    D_{KL}^{sym}\left( \mathcal{P}_1^{R_M, n} \middle| \middle| \mathcal{P}_2^{R_M, n} \right) \leq n \left( \frac{8 \alpha}{(1-2\alpha)^2} \right)^2 \norm{\rho_1 - \rho_2}_{Tr}^2
    \lesssim n \alpha^2 \norm{\rho_1 - \rho_2}_{Tr}^2.
\end{equation*}
The factor $n \alpha^2$ appearing in the above equation is the effective sample size which shows that the additional assumption of gentleness increases the necessary amount of copies needed for both state certification and tomography from $n \geq \Omega(\frac{1}{\epsilon^2})$ to $n \geq \Omega(\frac{1}{\epsilon^2 \alpha^2})$. Since gentleness and differential privacy are highly related, the bound in (\ref{eqn::main_result}) sheds a new light on existing bounds for quantum differentially-private mechanisms such as in \cite{hirche2023quantumdifferentialprivacyinformation}.

We further show that the theoretical limit for the distinguishability of the states by gentle measurements in Theorem \ref{thm::B3_C2} can be attained up to a constant in the general case. This attainability is the content of Theorem~\ref{thm::sharpness}. Furthermore, in Lemma~\ref{lem::improved_constant_positivity} we have seen that the relation between gentleness and quantum differential privacy may be improved for measurements whose operators are positive. In that special case, the results of Theorem~\ref{thm::B3_C2} can be improved such that inequality~\eqref{eqn::main_result} is asymptotically sharp for small $\alpha$ with the smaller constant $(4 \alpha)^2/(1- \alpha)^4$. We show that the qLS mechanism we propose attains this limit.

\begin{theorem}
\label{thm::sharpness}
    Let $\mathcal{H} = \mathbb{C}^d$, $\mathcal{S} = \mathcal{S}(\mathbb{C}^d)$ and $\alpha$ in $[0, 1]$. Then, there exists $M $ in $GM(\alpha, \mathcal{S})$ such that for any $\rho_1, \rho_2 \in \mathcal{S}$ it holds
    \begin{equation}
    \label{eqn::sharpness_equation}
        D_{KL}^{sym}(\mathcal{P}_{1}^{R_M} || \mathcal{P}_{2}^{R_M} ) \geq \left( \frac{4 \alpha}{1 + \alpha^2} \right)^2 \norm{\rho_1 - \rho_2}_{Tr}^2.
    \end{equation}
\end{theorem}

In order to show Theorem~\ref{thm::sharpness} we prove a gentle quantum Neyman-Pearson Lemma and derive bounds on the Kullback-Leibler-divergence based on this test. Let us fix a gentleness level $\alpha$ and set 
\begin{equation*}
    \delta = 4 \arctanh(\alpha).
\end{equation*}
For any two quantum states $\rho_0, \rho_1 \in \mathcal{S}(\mathbb{C}^d)$, we consider the difference $\rho_0 - \rho_1$. Since $\rho_0 - \rho_1$ is still hermitian with $\Tr(\rho_0 - \rho_1) = 0$, we can decompose it into a positive and negative part. Let $P_+$ and $P_-$ denote the projection onto the positive and negative part respectively. Note that $P_+ + P_- = \mathbbm{1}$. We define the two measurement operators
\begin{align}
\label{eqn::gentle_Neyman_Pearson_measurement}
    M_{\alpha, 1} = \sqrt{\frac{e^{\delta}}{e^{\delta}+1}}  P_+ + \sqrt{\frac{1}{e^{\delta}+1}} P_- \hspace{20pt} \text{and} \hspace{20pt} M_{\alpha, 0} = \sqrt{\frac{1}{e^{\delta}+1}}  P_+ + \sqrt{\frac{e^{\delta}}{e^{\delta}+1}} P_-
\end{align}
and combine them to a measurement $M_{\alpha} = (M_{\alpha, 0}, M_{\alpha, 1})$. These measurement operators are gentle-ized version of the Helstrom-measurement for quantum states \cite{Audenaert_2008}. Their $\alpha$-gentleness is a direct application of Lemma \ref{lem::qLS_gentleness}. The gentle quantum Neyman-Pearson test is the test $\Delta_{\alpha}^*$ that decides for $H_1$ if the measurement outcome of the above measurement defined in \eqref{eqn::gentle_Neyman_Pearson_measurement} is "1" and for $H_0$ if the measurement outcome is "0".
\begin{lemma}{\normalfont (Gentle quantum Neyman-Pearson test)}
\label{lem::gentle_Neyman_Pearson_test}
    Given the testing problem 
\begin{equation*}
    H_0: \rho = \rho_0 \hspace{20pt} \text{vs.} \hspace{20pt} H_1: \rho = \rho_1
\end{equation*}
the gentle quantum Neyman-Pearson test $\Delta_{\alpha}^*$ based on the measurement \eqref{eqn::gentle_Neyman_Pearson_measurement} has a total error given by
\begin{equation*}
    \mathcal{P}_0(\Delta_{\alpha}^* = 1) + \mathcal{P}_1(\Delta_{\alpha}^* = 0) = 1 - \frac{e^{\delta} -1 }{e^{\delta} + 1} \norm{\rho_0 - \rho_1}_{Tr} = 1 - \frac{2\alpha}{1 + \alpha^2} \norm{\rho_0 - \rho_1}_{Tr}.
\end{equation*}
\end{lemma}
\begin{proof}
    The total error of $\Delta_{\alpha}^*$ is then given by
\begin{align*}
    &\mathcal{P}_0(\Delta_{\alpha}^* = 1) + \mathcal{P}_1(\Delta_{\alpha}^* = 0) 
    \\
    =& \Tr\left[ \rho_0 \left(\frac{e^{\delta}}{e^{\delta}+1} P_+ + \frac{1}{e^{\delta}+1} P_- \right) \right] + \Tr\left[ \rho_1 \left(\frac{1}{e^{\delta}+1} P_+ + \frac{e^{\delta}}{e^{\delta}+1} P_- \right) \right]
    \\
    =& \frac{1}{e^{\delta} +1 } \left( (e^{\delta}-1) \Tr\left[ \rho_0P_+ + \rho_1P_- \right] + \Tr\left[ \rho_0 + \rho_1 \right]\right)
    \\
    =& \frac{1}{e^{\delta}+1} \left[ e^{\delta}-1 - (e^{\delta} -1) \norm{\rho_0 - \rho_1}_{Tr} + 2 \right]
    \\
    =& 1 - \frac{e^{\delta} -1 }{e^{\delta} + 1} \norm{\rho_0 - \rho_1}_{Tr},
\end{align*}
where we used that $\Tr\left[ \rho_0P_+ + \rho_1P_- \right] = 1- \norm{\rho_0 - \rho_1}_{Tr}$, which is the error of the non-gentle Neyman-Pearson test.
\end{proof}

\begin{proof}[Proof of Theorem~\ref{thm::sharpness}]
To show results of Theorem \ref{thm::sharpness} we now further study the error probabilities of the gentle quantum Neyman-Pearson test. We know that the optimal error between the two distributions $\mathcal{P}_0^{R_{M_{\alpha}}}$ and $\mathcal{P}_1^{R_{M_{\alpha}}}$ is given by
\begin{equation*}
    1 - \norm{\mathcal{P}_0^{R_{M_{\alpha}}} - \mathcal{P}_1^{R_{M_{\alpha}}}}_{TV}.
\end{equation*}
Since the error of our particular test is higher than the minimal error, we obtain the inequality
\begin{align*}
    1 - \norm{\mathcal{P}_0^{R_{M_{\alpha}}} - \mathcal{P}_1^{R_{M_{\alpha}}}}_{TV} \leq 1 - \frac{e^{\delta} -1 }{e^{\delta} + 1} \norm{\rho_0 - \rho_1}_{Tr},
    \intertext{which is equivalent to}
    \norm{\mathcal{P}_0^{R_{M_{\alpha}}} - \mathcal{P}_1^{R_{M_{\alpha}}}}_{TV} \geq \frac{e^{\delta} -1 }{e^{\delta} + 1} \norm{\rho_0 - \rho_1}_{Tr}.
\end{align*}
Using Pinsker's inequality yields
\begin{align*}
    \frac{1}{2} D_{KL}\left( \mathcal{P}_0^{R_{M_{\alpha}}} \middle| \middle| \mathcal{P}_1^{R_{M_{\alpha}}} \right) \geq \left( \frac{e^{\delta} -1 }{e^{\delta} + 1} \norm{\rho_0 - \rho_1}_{Tr} \right)^2.
\end{align*}
Since Pinsker's inequality still holds if we switch $\mathcal{P}_0^{R_{M_{\alpha}}}$ and $\mathcal{P}_1^{R_{M_{\alpha}}}$, we obtain the bound 
\begin{equation}
\label{eqn::sharpness}
    D_{KL}^{sym}\left( \mathcal{P}_0^{R_{M_{\alpha}}} \middle| \middle| \mathcal{P}_1^{R_{M_{\alpha}}} \right) \geq \left( \frac{4 \alpha}{1 + \alpha^2}\right)^2 \norm{\rho_0 - \rho_1}_{Tr}^2
\end{equation}
on the symmetrized Kullback-Leibler divergence when we resubstitute $\alpha$ for $\delta$. 
\end{proof}

Comparing the result of Theorem \ref{thm::B3_C2} to the one of Theorem~\ref{thm::sharpness} gives
\begin{equation*}
    \left( \frac{8 \alpha}{(1-2\alpha)^2} \right)^2 \norm{\rho_1 - \rho_2}_{Tr}^2 \geq D_{KL}^{sym}\left( \mathcal{P}_0^{R_{M_{\alpha}}} \middle| \middle| \mathcal{P}_1^{R_{M_{\alpha}}} \right) \geq \left( \frac{4 \alpha}{1 + \alpha^2}\right)^2 \norm{\rho_0 - \rho_1}_{Tr}^2.
\end{equation*}
Moreover, note that our gentle quantum Neyman-Pearson test is comprised of positive operators. As shown in Lemma \ref{lem::improved_constant_positivity} under the additional assumption of positivity of the measurement operators we have
\begin{equation*}
    \left( \frac{4 \alpha}{(1-\alpha)^2} \right)^2 \norm{\rho_1 - \rho_2}_{Tr}^2 \geq D_{KL}^{sym}\left( \mathcal{P}_0^{R_{M_{\alpha}}} \middle| \middle| \mathcal{P}_1^{R_{M_{\alpha}}} \right) \geq \left( \frac{4 \alpha}{1 + \alpha^2}\right)^2 \norm{\rho_0 - \rho_1}_{Tr}^2.
\end{equation*}
Thus, in the case of positive measurements we see that the upper and lower bound on the symmetrized Kullback-Leibler divergence only differs in the denominator. We can calculate that
\begin{equation*}
    {\left( \frac{4 \alpha}{1 + \alpha^2}\right)^2} \, / \,{\left( \frac{4 \alpha}{(1-\alpha)^2} \right)^2 } = \frac{(1-\alpha)^4}{(1 + \alpha^2)^2} = 1 + 4\alpha + O(\alpha^2)
\end{equation*}
which shows that, as $\alpha$ tends to 0, both upper and lower bound tend to the same value, proving the asymptotic sharpness of our qDPI and the gentle quantum Neyman-Pearson test under these constraints. In the general case, we see that the qLS optimally distinguishes quantum states under all gentle measurements up to a constant. Since, for $\alpha = 1$, the qLS is given by the optimal measurement in the non-gentle case, we believe it to be optimal for all $\alpha$. 

\section{Theoretical guarantees for gentle quantum learning}\label{Sec:Nec}
We now demonstrate how Theorem \ref{thm::B3_C2} increases the necessary amount of samples for quantum learning from $n \geq \Omega(\frac{1}{\epsilon^2})$ to $n \geq \Omega(\frac{1}{\epsilon^2 \alpha^2})$ under the constraint of local gentleness. 
\subsection{Gentle quantum state certification}
In quantum state certification the goal is to test the two hypotheses
\begin{equation}
\label{eqn::test_problem}
    H_0: \rho = \rho_0 \hspace{10pt} \text{vs.} \hspace{10pt} H_1(\epsilon): \norm{\rho - \rho_0}_{Tr} \geq \epsilon
\end{equation}
using a combination of quantum measurements and classical post-processing of the data obtained by the measurements. In particular, for a given measurement $M = (M_y)_{y \in \mathcal{Y}}$ and a subsequent test function $\Delta_n: \mathcal{Y} \to \{0,1\}$, we are interested in the sum of errors
\begin{equation*}
    \mathcal{P}_e = \mathcal{P}_{\rho_0^{\otimes n}}(\Delta_n = 1) + \sup_{\rho \in H_1(\epsilon)} \mathcal{P}_{\rho^{\otimes n}}(\Delta = 0)
\end{equation*}
and its optimum for all possible test
\begin{equation*}
    \mathcal{P}_e^* = \inf_{(M_{\alpha}, \Delta_n)} \mathcal{P}_e.
\end{equation*}
Here, the infimum is over all possible pairs of locally-$\alpha$-gentle measurements $M_{\alpha}$ and subsequent test functions~$\Delta_n$. We will make use of the qDPI in Theorem \ref{thm::B3_C2} to show a lower bound on the number $n$ of copies needed. 
\begin{lemma}
    To solve the state certification problem (\ref{eqn::test_problem}) with probability $2/3$ under the constraint of local $\alpha$-gentleness, that is $\mathcal{P}_e^* \leq 1/3$, a total of $n \geq \Omega(\frac{1}{\epsilon^2 \alpha^2})$ copies of the state $\rho$ are needed.
\end{lemma}

\begin{proof}
We will show that for locally-$\alpha$-gentle measurements $M = M_1 \otimes ... \otimes M_n$, the distinguishability of the outcome distributions is decreased by the constraint of gentleness, such that every classical post-processing of the data needs an increased number of copies to solve the problem within a certain error. To do that, we see that when we apply an arbitrary locally-gentle measurement $M = M_1 \otimes ... \otimes M_n$ to the state $\rho^{\otimes n}$ we obtain a random outcome $y = (y_1,...,y_n)$. The distribution of the outcome $y$ is given by $\mathcal{P}_{\rho^{\otimes n}}^{R_M} = \mathcal{P}_{\rho}^{R_{M_1}} \otimes ... \otimes \mathcal{P}_{\rho}^{R_{M_n}}$. Based on this random outcome, we now chose to either reject or accept the null-hypothesis. This decision is given by the test function $\Delta_n$. We know that the optimal test-function $\Delta_n^*$ has an error that is given by
\begin{align}
\label{eqn::max_test_error}
    \mathcal{P}_e^* = \mathcal{P}_{\rho_0^{\otimes n}}^{R_M}(\Delta_n^* = 1) + \sup_{\rho \in H_1} \mathcal{P}_{\rho^{\otimes n}}^{R_M}(\Delta_n^* = 0).
\end{align}
It is easy to see that this error is larger than the error for a fixed $\rho_1 \in H_1$ with $\norm{\rho_1 - \rho_0}_{Tr} = \epsilon$. We write, using Pinsker's inequality,
\begin{align*}
    \frac{1}{3} \geq \mathcal{P}_e^* &\geq \mathcal{P}_{\rho_0^{\otimes n}}^{R_M}(\Delta_n^* = 1) + \mathcal{P}_{\rho_1^{\otimes n}}^{R_M}(\Delta_n^* = 0) \\
    &\geq 1 - \sqrt{\frac{1}{4} D_{KL}^{sym}\left(\mathcal{P}_{\rho_0^{\otimes n}}^{R_M} \middle| \middle|\mathcal{P}_{\rho_1^{\otimes n}}^{R_M} \right)}
    \geq 1 - \sqrt{\frac{1}{4} \sum_{i = 1}^n D_{KL}^{sym}\left( \mathcal{P}_{\rho_0}^{R_{M_i}} \middle| \middle| \mathcal{P}_{\rho_1}^{R_{M_i}} \right)}
    \\
    &\geq 1 - \frac{1}{2} \sqrt{n} \frac{8\alpha}{(1-2\alpha)^2} \norm{\rho_0 - \rho_1}_{Tr} 
    = 1 - \frac{4}{(1-2\alpha)^2} \sqrt{n} \cdot \alpha \epsilon
\end{align*}
which shows that in order to keep the total error below $1/3$ a total of
\begin{equation*}
    n \geq \left( \frac{(1-2\alpha)^2}{12} \right)^2 \frac{1}{\epsilon^2 \alpha^2}
\end{equation*}
copies of the state $\rho$, i.e. $n \geq \Omega(\frac{1}{\epsilon^2\alpha^2})$, are needed for gentle state certification. 
\end{proof}

\subsection{Gentle quantum tomography}
\label{sec::Tomography_lower_bound}

The question in quantum tomography is how many copies of an unknown state $\rho$ are needed to estimate it in trace norm up to a given error $\epsilon$. That is, we want to measure a given state $\rho^{\otimes n}$ using a gentle measurement $M = (M_y)_{y \in \mathcal{Y}}$ and give an estimator $\hat{\rho}_n: \mathcal{Y} \to \mathcal{S}(\mathbb{C}^2)$ such that its mean-squared-error (MSE)
\begin{equation*}
    \sup_{\rho} \mathbb{E}_{\rho}\left[ \norm{\hat{\rho}_n - \rho}_{Tr}^2 \right]
\end{equation*}
is at most $\epsilon^2$. In particular, we consider the minimax risk of this problem
\begin{equation*}
    \mathcal{R}^* = \inf_{(M_{\alpha}, \hat{\rho}_n)} \sup_{\rho} \mathbb{E}_{\rho}\left[ \norm{\hat{\rho}_n - \rho}_{Tr}^2 \right].
\end{equation*}
Here, the infimum is taken over all locally-$\alpha$-gentle measurement $M_{\alpha}$ and subsequent estimators $\hat{\rho}_n$. We will see that this problem is more difficult than the problem of quantum state certification. As such, we show a lower bound on the copies needed for tomography by relating it to the lower bound for state certification.
\begin{lemma}
    A total of $n \geq \Omega(\frac{1}{\epsilon^2 \alpha^2})$ copies of the state $\rho$ are needed in order to learn it up to an error $\epsilon$ in trace-norm. That is, in order to have $\mathcal{R}^* \leq \epsilon^2$, we need $n \geq \Omega(\frac{1}{\epsilon^2 \alpha^2})$.
\end{lemma}
\begin{proof}
As in state certification, we apply a gentle measurement $M = M_1 \otimes ... \otimes M_n$ to the state $\rho^{\otimes n}$ and obtain a random outcome $y = (y_1,...,y_n)$. For every possible random outcome $y$  we give an estimate $\hat{\rho}_n(y)$ of the state. The maximal error of this estimator is given by
\begin{equation*}
    \sup_{\rho} \mathbb{E}_{\rho}\left[ \norm{\hat{\rho}_n - \rho}_{Tr}^2 \right].
\end{equation*}
By the Markov inequality we can bound this expression from below for any $\epsilon > 0$ and a given fixed state $\rho_0$ by
\begin{align*}
    \frac{4}{\epsilon^2} \sup_{\rho} \mathbb{E}_{\rho}\left[ \norm{\hat{\rho}_n - \rho}_{Tr}^2 \right] &\geq \sup_{\rho} \mathcal{P}_{\rho}\left( \norm{\hat{\rho}_n - \rho}_{Tr} \geq \frac{\epsilon}{2} \right)
    \\
    &= \max\left\{ \mathcal{P}_{\rho_0}\left( \norm{\hat{\rho}_n - \rho_0}_{Tr}^2 \geq \frac{\epsilon}{2} \right), \sup_{\rho \neq \rho_0}\mathcal{P}_{\rho}\left( \norm{\hat{\rho}_n - \rho}_{Tr} \geq \frac{\epsilon}{2} \right) \right\}
    \\
    &\geq \frac{1}{2} \left( \mathcal{P}_{\rho_0}\left( \norm{\hat{\rho}_n - \rho_0}_{Tr} \geq \frac{\epsilon}{2} \right) + \sup_{\norm{\rho - \rho_0}_{Tr} \geq \epsilon} \mathcal{P}_{\rho}\left( \norm{\hat{\rho}_n - \rho}_{Tr} \geq \frac{\epsilon}{2} \right) \right)
    \intertext{Now, if $\norm{\rho - \rho_0} \geq \epsilon$ and $\norm{\hat{\rho}_n - \rho} < \frac{\epsilon}{2}$, then $\norm{\hat{\rho}_n - \rho_0} \geq \frac{\epsilon}{2}$, which further leads to}
    &\geq \frac{1}{2} \left( \mathcal{P}_{\rho_0}\left( \norm{\hat{\rho}_n - \rho_0}_{Tr} \geq \frac{\epsilon}{2} \right) + \sup_{\norm{\rho - \rho_0}_{Tr} \geq \epsilon} \mathcal{P}_{\rho}\left( \norm{\hat{\rho}_n - \rho_0}_{Tr} < \frac{\epsilon}{2} \right) \right)
    \\
    &= \frac{1}{2} \left( \mathcal{P}_{\rho_0}\left( \Delta = 1 \right) + \sup_{\norm{\rho - \rho_0}_{T r} \geq \epsilon} \mathcal{P}_{\rho}\left( \Delta = 0 \right) \right),
    \intertext{where $\Delta$ is the indicator-function on $\{ \norm{\hat{\rho}_n - \rho_0}_{Tr}^2 \geq \frac{\epsilon}{2} \}$. We can therefore bound the estimation error by the testing error (\ref{eqn::max_test_error}) for the testing problem (\ref{eqn::test_problem}) by}
    \sup_{\rho} \mathbb{E}_{\rho}\left[ \norm{\hat{\rho}_n - \rho}_{Tr}^2 \right] &\geq \frac{\epsilon^2}{8} \left( 1 - \frac{4}{(1-2\alpha)^2} \sqrt{n} \alpha \epsilon \right).
    \intertext{Chosing 
    \begin{equation*}
        \epsilon = \min \left\{1, \frac{(1-2\alpha)^4}{6 \alpha \sqrt{n}} \right\}
    \end{equation*}
    gives the final lower bound}
    \sup_{\rho} \mathbb{E}_{\rho}\left[ \norm{\hat{\rho}_n - \rho}_{Tr}^2 \right] &\geq \min \left\{ \frac{1}{8}, \frac{(1-2\alpha)^4}{864 n \alpha^2} \right\}.
\end{align*}
Since this final lower bound is independent of the actual estimator $\hat{\rho}_n$, it holds for all estimators, showing
\begin{equation*}
    \mathcal{R}^* \geq \min \left\{ \frac{1}{8}, \frac{(1-2\alpha)^4}{864 n \alpha^2} \right\}.
\end{equation*}
We see, that a total of $n \geq \Omega(\frac{1}{\epsilon^2\alpha^2})$ copies are needed to learn the state $\rho$ up to an error $\epsilon$ in trace-norm.
\end{proof}

\section{Attainability: gentle measurement techniques}
\label{sec::quantum_label_switch}
In the preceding section we have seen how we may use our qDPI to derive lower bounds of $n \geq \Omega(\frac{1}{\epsilon^2\alpha^2})$ on the amount of copies needed for gentle quantum state certification and tomography. We will now show that a sample size of $n = O(\frac{1}{\epsilon^2 \alpha^2})$ is not only necessary but also sufficient under those constraints. To do this, we demonstrate how gentle measurement techniques for qubits can be implemented using existing tools of quantum computing. 

\subsection{The quantum Label Switch measurement}
Suppose we want to measure a quantum state $\rho$ gently in the computational bases $\ket{0}, \ket{1}$. The typical measurement $M = (M_0, M_1)$ for rank-1 matrices
\begin{equation*}
    M_0 = \ket{0}\bra{0} \hspace{20pt} \text{and} \hspace{20pt} M_1 = \ket{1}\bra{1}
\end{equation*}
is not gentle for any $\alpha \in [0, 1)$ as we have seen in Section \ref{sec::properties}. In order to gentle-ize this measurement, we draw inspiration from classical differential privacy using the connection between the two concepts shown by \cite{aaronson2019gentle}. Let us specify a gentleness level $\alpha$ and again set 
\begin{equation*}
    \delta = 4 \arctanh(\alpha).
\end{equation*}
We then prepare a second ancillary state 
\begin{equation*}
    \sigma_{\delta} = \ket{\varphi_{\delta}}\bra{\varphi_{\delta}}, \quad \text{where }
\ket{\varphi_{\delta}} = \sqrt{\frac{e^{\delta}}{e^{\delta}+1}} \ket{0} + \sqrt{\frac{1}{e^{\delta}+1}}\ket{1}.
\end{equation*}
We then consider the product state $\rho \otimes  \sigma_{\delta}$ and pass it through a CNOT-gate, that is
\begin{equation*}
    \rho_{\delta} = \text{CNOT}(\rho \otimes  \sigma_{\delta})\text{CNOT}^*.
\end{equation*}
Measuring the second register of this state $\rho_{\delta}$ in the computational basis, that is by using the measurement $M = (\text{Id} \otimes \ket{y}\bra{y})_{y \in \{ 0,1 \}}$, disturbs the first register only by a slight amount, namely by $\alpha$, \textit{i.e.} $\norm{\rho - \rho_{M \to y}}_{Tr} \leq \alpha$ for $y = 0,1$. This works because the CNOT gate entangles the state $\rho \otimes \sigma_{\delta}$ in such a way that the measurement of the second register has only a small impact on the state of the first register. As a consequence, the information extracted from the system is dependent on $\alpha$ as well and decreases with decreasing $\alpha$. In comparison to standard tomography protocols for qubits \cite{Schmied_2016}, our mechanism requires both the preparation of the ancillary state as well as its entanglement with the original state. The measurement afterwards is the same for both the standard and the gentle tomography. Due to the locality of the procedure however, the ancillary register may be reused for multiple states, limiting the additional costs when compared to standard tomography.

Mathematically, the qLS measurement technique we just described can be written in the form of Definition~\ref{defn::quantum_measurement} as $M_{\alpha} = (M_{\alpha, 0}, M_{\alpha,1})$ where
\begin{equation}
    \label{eqn::label_switch_mathematical_formulation}
    M_{\alpha, 0} = \sqrt{\frac{e^{\delta}}{e^{\delta}+1}} \ket{0}\bra{0} + \sqrt{\frac{1}{e^{\delta}+1}} \ket{1}\bra{1} \text{ and } 
    M_{\alpha, 1} = \sqrt{\frac{1}{e^{\delta}+1}} \ket{0}\bra{0} + \sqrt{\frac{e^{\delta}}{e^{\delta}+1}} \ket{1}\bra{1}.
\end{equation}

This gentle-ization technique derives its name from the well known privatization mechanism for Bernoulli random variables \cite{steinberger2023efficiency}. Compared to non-gentle measurement $M = (M_0, M_1)$
we ``swap'' the outcomes with a probability of $\frac{1}{e^{\delta}+1}$ and keep the true outcome with a probability of $\frac{e^{\delta}}{e^{\delta}+1}$.
\begin{lemma}
    The qLS measurement defined in (\ref{eqn::label_switch_mathematical_formulation}) with $\delta = 4 \arctanh(\alpha)$ is $\alpha$-gentle on all qubits, i.e. $M_{\alpha} \in GM(\alpha, \mathcal{S}(\mathbb{C}^2))$.
\end{lemma}
This result is a special case of Lemma \ref{lem::qLS_gentleness} for $d = 2$. In the following sections we study the statistical properties of the resulting outcome probabilities. We use the same methodology of label switching for both state certification and tomography in the gentle regime.

\subsection{Gentle quantum tomography of qubits}
\label{sec::Tomography_upper_bound}
A well known and easy tomography method for qubits is the ``scaled direct inversion'' \cite{Schmied_2016}. Any qubit can be represented in the Pauli basis as $\rho = \rho(\mathbbm{r}) = \frac{1}{2}\left( \text{Id} + \mathbbm{r}_x \sigma_x + \mathbbm{r}_y \sigma_y+ \mathbbm{r}_z \sigma_z \right)$ where $\mathbbm{r}$ is a real vector of length less than or equal to 1. The scaled inversion method aims at learning the vector $\mathbbm{r}$ by estimating the three axes individually using basis measurements for each axis, that is, measuring $M_x = (\ket{x_+}\bra{x_+}, \ket{x_-}\bra{x_-})$, $M_y = (\ket{y_+}\bra{y_+}, \ket{y_-}\bra{y_-})$ and $M_z = (\ket{z_+}\bra{z_+}, \ket{z_-}\bra{z_-})$. A gentle version of this measurement is given by applying the quantum Label Switch to each of the three basis measurements. We obtain the following result that is an upper bound on number of copies sufficient for tomography.

\begin{lemma}
\label{lem::gentle_tomography}
    There exists a locally-$\alpha$-gentle measurement  $M_{\alpha}$ using the qLS mechanism and an estimator $\hat{\rho}_n$ based on the outcomes of the measurement $M_{\alpha}$ that can learn a state $\rho$ up to an error of $\epsilon$ in trace-norm using $n = O(\frac{1}{\epsilon^2 \alpha^2})$ copies of $\rho$.
\end{lemma}

\begin{proof}
As in the non-gentle case, we consider each axis separately. For the $z$-Axis, we consider the measurement $M_{\alpha, z} = (M_{\alpha, z_+}, M_{\alpha, z_-})$ given by 
\begin{align}
    M_{\alpha, z_+} =& \sqrt{\frac{e^{\delta}}{e^{\delta}+1}} \ket{z_+}\bra{z_+} + \sqrt{\frac{1}{e^{\delta}+1}} \ket{z_-}\bra{z_-} \notag
    \\ 
    &\text{and} \label{eqn::definition_qLS_z_direction}
    \\
    M_{\alpha, z_-}  = & \sqrt{\frac{1}{e^{\delta}+1}} \ket{z_+}\bra{z_+} + \sqrt{\frac{e^{\delta}}{e^{\delta}+1}} \ket{z_-}\bra{z_-} \notag
\end{align}
for $\delta = 4 \arctanh(\alpha)$. Its gentleness on $\mathcal{S}(\mathbb{C}^2)$ is assured by Lemma \ref{lem::qLS_gentleness}. We can actually calculate the post-measurement state for an arbitrary input state $\rho(\mathbbm{r})$. Although this description does not aid in the calculation of its gentleness, we state it in Lemma \ref{lem::post_measurement_state_qls} in the supplementary material. When measuring a state $\rho(\mathbbm{r})$ using the measurement $M_{\alpha, z}$ we obtain outcome $z_+$ with probability
\begin{align*}
        \mathcal{P}_{\mathbbm{r}}\left( R^{M_{\alpha, z}} = z_+ \right) =& \Tr\left[ \rho(\mathbbm{r}) M_{\alpha, z_+} M_{\alpha, z_+} \right]
        \\
        =& \left( \frac{e^{\delta}}{e^{\delta}-1} \Tr[\ket{{z_+}}\bra{{z_+}} \rho(\mathbbm{r})] + \frac{1}{e^{\delta}-1} \Tr[\ket{{z_-}}\bra{{z_-}} \rho(\mathbbm{r})] \right)
        \\
        =& \frac{1}{2} \frac{1}{e^{\delta}-1} \bigg( e^{\delta} \bra{{z_+}}\mathbbm{1}_{\mathbb{C}^2}+\mathbbm{r}_x\sigma_x + \mathbbm{r}_y\sigma_y + \mathbbm{r}_z\sigma_z \ket{{z_+}} 
        \\
        &+ \bra{{z_-}}\mathbbm{1}_{\mathbb{C}^2}+\mathbbm{r}_x\sigma_x + \mathbbm{r}_y\sigma_y + \mathbbm{r}_z\sigma_z \ket{{z_-}} \bigg)
        \\
        =& \frac{1}{2} \frac{1}{e^{\delta}-1} \Big( e^{\delta}(1+\mathbbm{r}_z) + 1 - \mathbbm{r}_z \Big)
        \\
        =& \frac{1}{2} + \frac{1}{2} \frac{e^{\delta}-1}{e^{\delta}+1} \mathbbm{r}_z =: p_{z_+}.
        \intertext{Analogously, we find}
        \mathcal{P}_{\mathbbm{r}}\left( R^{M_{\alpha, z}} = z_- \right) =& \frac{1}{2} - \frac{1}{2} \frac{e^{\delta}-1}{e^{\delta}+1} \mathbbm{r}_z =: p_{z_-}.
\end{align*}
The expectation of the random variable $R^{M_{\alpha, z}}$ is given by
\begin{equation*}
    \mathbb{E}_{\mathbbm{r}}\left[ R^{M_{\alpha, z}} \right] = p_{z_+} - p_{z_-} = \frac{e^{\delta}-1}{e^{\delta}+1} \mathbbm{r}_z.
\end{equation*}
When measuring $M_{\alpha, z}$ $n$times on $n$ identical copies of $\rho(\mathbbm{r})$ we obtain the outcome $z_+$ $n_{z_+}$-times and $z_-$ is measured $n_{z_-} = n - n_{z_-}$-times. Using the linearity of the expected value, we can see that the expectation of the difference in relative frequencies is given as
\begin{equation*}
    \mathbb{E}_{\mathbbm{r}}\left[ \frac{n_{z_+} - n_{z_-}}{n} \right] = \frac{1}{n} \mathbb{E}_{\mathbbm{r}}\left[ n_{z_+} - n_{z_-} \right] = p_{z_+} - p_{z_-} = \frac{e^{\delta}-1}{e^{\delta}+1} \mathbbm{r}_z.
\end{equation*}
The factor $({e^{\delta}+1})/({e^{\delta}-1}) = ({\alpha^2+1})/({2\alpha} )$ appearing is due to gentleness. It introduces a bias whose correction is the reason for the additional factor of $\alpha^2$ appearing in the bound. Let us consider
\begin{equation*}
    \Tilde{\mathbbm{r}}_{n,z} = \frac{e^{\delta}+1}{e^{\delta}-1} \frac{n_{z_+} - n_{z_-}}{n}.
\end{equation*}
$\Tilde{\mathbbm{r}}_{n,z}$ is then an unbiased estimator of $\mathbbm{r}_z$. Its mean squared error is given by the variance
\begin{align*}
    \mathbb{E}_{\mathbbm{r}}\left[ \norm{ \Tilde{\mathbbm{r}}_{n,z} - \mathbbm{r}_z }_2^2\right] = \left( \frac{e^{\delta}+1}{e^{\delta}-1} \right)^2 \Var_{\mathbbm{r}}\left[\frac{n_{z_+} - n_{z_-}}{n}\right] &= \left(\frac{e^{\delta}+1}{e^{\delta}-1}\right)^2 \Var_{\mathbbm{r}}\left[ \frac{2 n_{z_+} - n}{n} \right] 
    \\
    &= \frac{4 (e^{\delta}+1)^2}{n (e^{\delta}-1)^2} \, p_{z_+}p_{z_-}
    \\
    &\leq \frac{(e^{\delta}+1)^2}{n (e^{\delta}-1)^2}.
\end{align*}
When we define estimators of $\mathbbm{r}_x$ and $\mathbbm{r}_y$ in an analogous way, we obtain an estimator of the full parameter vector
\begin{equation*}
    \Tilde{\mathbbm{r}}_n = \left[ \Tilde{\mathbbm{r}}_{n,x}, \Tilde{\mathbbm{r}}_{n,y}, \Tilde{\mathbbm{r}}_{n,z} \right].
\end{equation*}
The $\ell_2$-error is given by
\begin{align*}
    \mathbb{E}_{\mathbbm{r}}\left[ \norm{\Tilde{\mathbbm{r}}_n - \mathbbm{r}}_2^2 \right] &= \mathbb{E}_{\mathbbm{r}}\left[ \norm{ \Tilde{\mathbbm{r}}_{n,x} - \mathbbm{r}_x }_2^2\right] + \mathbb{E}_{\mathbbm{r}}\left[ \norm{ \Tilde{\mathbbm{r}}_{n,y} - \mathbbm{r}_y }_2^2\right] + \mathbb{E}_{\mathbbm{r}}\left[ \norm{ \Tilde{\mathbbm{r}}_{n,z} - \mathbbm{r}_z }_2^2\right]
    \\
    &= \frac{4 (e^{\delta}+1)^2}{n (e^{\delta}-1)^2} \left( p_{x_+}p_{x_-} + p_{y_+}p_{y_-} + p_{z_+}p_{z_-}  \right)
    \\
    &\leq \frac{3 (e^{\delta}+1)^2}{n (e^{\delta}-1)^2}.
\end{align*}
Since $\norm{\Tilde{\mathbbm{r}}_n}_2$ is not necessarily smaller than 1, $\Tilde{\mathbbm{r}}_n$ is not necessarily the parameter vector of a valid quantum state. To remedy this problem, we scale the estimator down whenever it is too large. Geometrically this is the projection of the (not-necessarily positive) matrix $\rho(\Tilde{\mathbbm{r}}_n)$ onto the convex subspace of quantum states in $\mathbb{C}^2$. For qubits, this can be easily calculated as
\begin{equation*}
    \hat{\rho}_n = \rho(\hat{\mathbbm{r}}_n) \hspace{10pt} \text{where } \hat{\mathbbm{r}}_n = \frac{\Tilde{\mathbbm{r}}_n}{\norm{\Tilde{\mathbbm{r}}_n}_2 \vee 1 }.
\end{equation*}
Since $\hat{\rho}_n = \rho(\hat{\mathbbm{r}}_n)$ is a projection onto the convex set of states, its error is smaller than the error of the non-projected $\rho(\Tilde{\mathbbm{r}}_n)$. Using Lemma \ref{lem::trace_norm_qubits}, we obtain
\begin{align}
    \mathbb{E}_{\mathbbm{r}}\left[ \norm{\rho(\hat{\mathbbm{r}}_n) - \rho(\mathbbm{r})}_{Tr}^2 \right] &= \frac{1}{4} \mathbb{E}_{\mathbbm{r}}\left[ \norm{\hat{\mathbbm{r}}_n - \mathbbm{r}}_2^2 \right] \notag
    \\
    \label{eqn::error_tomography_qubits}
    &\leq \frac{1}{4} \mathbb{E}_{\mathbbm{r}}\left[ \norm{\Tilde{\mathbbm{r}}_n - \mathbbm{r}}_2^2 \right] \leq \frac{3 (e^{\delta}+1)^2}{4 n (e^{\delta}-1)^2} = \frac{3(\alpha^2 +1)^2}{16 \alpha^2 n}. 
\end{align}
This error bound shows that in order to learn an arbitrary state $\rho(\mathbbm{r})$ up to an an error of $\epsilon$ in trace-norm. A total of of $n = O(\frac{1}{\epsilon^2 \alpha^2})$ copies of the state are sufficient. In Section \ref{sec::Tomography_lower_bound} we have seen that this loss is unavoidable, which shows that the proposed quantum Label Switch mechanism is optimal for gentle state tomography.
\end{proof}

\subsection{Gentle quantum state certification}
The quantum Label Switch is also suitable for gentle state certification. That is, we may use the same measurement technique to derive a tool for solving gentle state certification. 
\begin{lemma}
\label{lem::gentle_state_certification}
    There exists a locally-$\alpha$-gentle measurement $M_{\alpha}$ using the qLS mechanism that can solve the state certification problem (\ref{eqn::test_problem}) with probability $2/3$ using $n = O(\frac{1}{\epsilon^2 \alpha^2})$ copies of $\rho$.
\end{lemma}
\begin{proof}
Suppose we are given $n$ identical copies of the unknown state $\rho = \rho(\mathbbm{r})$ and want to know whether $\rho$ is equal to a fixed reference state $\rho_0 = \rho(\mathbbm{r}_0)$ or $\epsilon$ far away from it in trace-norm for a given $\epsilon > 0$ under the constraint of gentleness. One way to do this, is by measuring the state gently and obtain an estimator $\hat{\rho}_n = \rho(\hat{\mathbbm{r}}_n)$ of $\rho$ as we did for tomography in Section \ref{sec::Tomography_upper_bound}. We decide $H_0$ if $\norm{\hat{\rho}_n - \rho_0}_{Tr} \leq \frac{\epsilon}{2}$ and $H_1$ if $\norm{\hat{\rho}_n - \rho_0}_{Tr} > \frac{\epsilon}{2}$. The total error is then given by
\begin{align*}
    \mathcal{P}_{\rho_0}\left(\norm{\hat{\rho}_n - \rho_0}_{Tr} > \frac{\epsilon}{2}\right) + \sup_{\rho: \norm{\rho - \rho_0}_{Tr} > \epsilon} \mathcal{P}_{\rho}\left(\norm{\hat{\rho}_n - \rho_0}_{Tr} \leq \frac{\epsilon}{2}\right).
\end{align*}
Now, by the triangle-inequality, if $\norm{\rho - \rho_0}_{Tr} > \epsilon$ and $\norm{\hat{\rho}_n - \rho_0}_{Tr} \leq \epsilon/2$, then $\norm{\hat{\rho}_n - \rho} > \epsilon/2$. This allows us to upper bound the error from above by
\begin{align*}
    \mathcal{P}_{\rho_0}\left(\norm{\hat{ \rho}_n - \rho_0}_{Tr} > \frac{\epsilon}{2} \right) + \sup_{\rho: \norm{\rho - \rho_0}_{Tr} > \epsilon} \mathcal{P}_{\rho}\left(\norm{\hat{\rho}_n - \rho}_{Tr} > \frac{\epsilon}{2} \right).
\end{align*}
The Markov inequality lets us further bound this term from above by
\begin{align*}
    \frac{4}{\epsilon^2} \mathbb{E}_{\rho_0}\left[ \norm{\hat{\rho}_n - \rho_0}_{Tr}^2 \right] + \frac{4}{\epsilon^2} \sup_{\rho: \norm{\rho - \rho_0}_{Tr} > \epsilon} \mathbb{E}_{\rho}\left[ \norm{\hat{\rho}_n - \rho}_{Tr}^2 \right] \leq \frac{8}{\epsilon^2} \sup_{\rho} \mathbb{E}_{\rho}\left[ \norm{\hat{\rho}_n - \rho}_{Tr}^2 \right]
\end{align*}
We see that this last term is the error of the tomography problem given by equation~\eqref{eqn::error_tomography_qubits}. This gives us the final upper bound
\begin{equation*}
    \mathcal{P}_{\rho_0}\left(\norm{\hat{\rho}_n - \rho_0}_{Tr} > \frac{\epsilon}{2}\right) + \sup_{\rho: \norm{\rho - \rho_0}_{Tr} > \epsilon} \mathcal{P}_{\rho}\left(\norm{\hat{\rho}_n - \rho_0}_{Tr} \leq \frac{\epsilon}{2}\right) \leq \frac{3 (1+\alpha^2)^2}{2} \frac{1}{n \epsilon^2 \alpha^2}.
\end{equation*}
This shows that for state tomography, a total of $n = O\left( \frac{1}{\epsilon^2 \alpha^2} \right)$ copies of $\rho$ are needed to bound the right hand side by $1/3$, again showing the additional factor of $\frac{1}{\alpha^2}$ appearing.
\end{proof}

\section{Discussion}
We have studied quantum statistical learning of qubits under the effects of gentle measurements. The main result of this paper is the strong quantum data processing inequality Theorem~\ref{thm::B3_C2} that relates the symmetrized Kullback-Leibler-divergence of the probability distributions of two states from a gentle measurement to the trace distance of the states and the gentleness parameter $\alpha$. This result allows us to show a necessary increase of the effective sample size for state tomography and state certification from $\frac{1}{ \epsilon^2 }$ to $\frac{1}{ \epsilon^2 \alpha^2}$. By comparison with our gentle Neyman-Pearson test, we show that it is sharp in the high-gentleness regime, that is $\alpha$ being small. Furthermore, by introducing the quantum label switch (qLS) we propose a physically implementable mechanism that attains the rate from the lower bound. 

\subsection{From local to global gentleness}
Having explored the local impact of locally-gentle measurements, it is a natural next step to study the global impact of gentle measurements, in the sense that we do not view $\rho^{\otimes n}$ as $n$ individual systems but as one singular enlarged system. A local (product) measurement that perturbs every subsystem naturally perturbs the whole system even more. The question is whether or not this global perturbation can also be limited by a local measurement. The answer to this question is no. Let us assume that such a global $\alpha_g$-gentle product measurement $M = M_1 \otimes... \otimes M_n$ exists. Lemma~\ref{lemma::3} assures us that it suffices to study the gentleness of this measurement on pure states only. By Lemma~\ref{lemmaTraceNorm} we then know that each $M_i$ is $\alpha_l$-gentle on its respective register where
\begin{equation}
\label{eqn::relation_globa_local_gentleness}
        \alpha_g = \sqrt{1 - (1 - \alpha_l^2)^n} .
\end{equation}
For $\alpha_g$ to stay bounded away from $1$, we need $\alpha_l \lesssim n^{-\frac{1}{2}}$. The result of Theorem \ref{thm::sharpness} then shows that the measurement $M$ cannot be globally-$\alpha_g$-gentle on product states while being informative for any fixed $\alpha_g < 1$, in the sense that there exists an estimator whose error tends to $0$ as $n$ grows. In order to achieve global gentleness, one must apply a global (product) measurement such as the Laplace-mechanism in \cite{aaronson2019gentle}, which is inspired by the classical Laplace mechanism in differential privacy \cite{duchi2014localprivacydataprocessing}. This mechanism however also requires access to- and manipulation of the system as a whole which is currently physically unfeasible.

\begin{table}[ht]
\renewcommand{\arraystretch}{1.8} 
    \centering
    \begin{tabular}{|c|c|c|}
        \hline
        \diagbox[width=12em, height=4em]{System}{Measurement} & Global/entangled & Local/unentangled \\
        \hline
        Global & Laplace mechanism \cite{aaronson2019gentle} &  Impossible by~\eqref{eqn::relation_globa_local_gentleness} \\
        \hline
        Local & - & qLS in Lemmas \ref{lem::gentle_state_certification}, \ref{lem::gentle_tomography} \\
        \hline
    \end{tabular}
    \caption{\justifying An overview of which measurements achieve local or global gentleness on a system of $n$ identical states $\rho^{\otimes n}$. In the local case we view the system as composed of $n$ individual states $\rho_1,...,\rho_n$ and study the gentleness on each singular system. In the global case we view the system as one large system $\rho_1 \otimes ... \otimes \rho_n$ and study the global gentleness on the whole system. We do not study the impact of global measurements on local systems because a local system ceases to be local as soon as we entangle it with the others.}
    \label{tab::Gentleness_per_states}
\end{table}

\subsection{From qubits to qudits}
Although qubits form the building blocks of many quantum systems, most importantly quantum computers, it would be interesting to study the properties of gentle measurements on higher dimensional quantum states. Most of the results we presented here are equally valid for $d$-dimensional as for $2$-dimensional states. However, as is the case for classical inference under local differential privacy, we expect the optimal rates to be dependent on the dimension $d$. However, deriving such results would require developing a far deeper understanding of gentle measurements than what has already been done. We believe it to be an interesting topic for future research to study the effects of gentle measurements on general $d$-dimensional quantum states. 

Another possible future topic of study are adaptive sequential gentle measurements, which will definitely improve upon the results of the qLS mechanism. It remains open whether this improvement will be substantial in that it changes the dependence of the rate on $\alpha$. The authors of this paper conjecture that the minimax optimal rate for both state tomography and certification will be $\Theta(\frac{1}{\epsilon^2 \alpha^2})$.

As with classical differential privacy, where a lot of advancements have been made in recent years \cite{acharya_inference_2021, Butucea2020, Butucea2023, Cai2021, duchi2014localprivacydataprocessing, Rohde2020, Wasserman2010}, one can consider even more involved problems like estimating nonparametric quantum states or quadratic functionals of such. However one must keep in mind that although connected, differential privacy and gentleness are still distinctive concepts that require different treatments. For example, one cannot simply project a state onto a lower dimensional subspace gently. As such, typical projection estimators cannot easily be adapted to gentleness. While the exact implications of gentleness on infinite dimensional systems are therefore still unexplored, we believe it to still reduce the effective sample size in those regimes. As such, we there are many interesting discoveries to be made for gentle measurements.

\section{Additional proofs and calculations}
\label{sec::additioinal_proofs_and_calculations}
\subsection{Explicit calculation of the post-measurement state under the qLS-mechanism.}
Let us give an explicit calculation of the post-measurement state of a qubit after the measurement of the qLS mechanism in the $z$-direction. The post-measurement state in the $x$ and $y$ directions can be calculated analogously by permuting the corresponding parameters. The calculations are straightforward, but long.

\begin{lemma} 
\label{lem::post_measurement_state_qls}
For the post measurement states of a state $\rho(\mathbbm{r})$ under the measurement $M_{\alpha, z}$ defined in~\eqref{eqn::definition_qLS_z_direction}, it holds
    \begin{align*}
        \rho(\mathbbm{r})_{M_{\alpha, z} \to z_+} = \rho(\mathbbm{r}_{M_{\alpha, z} \to z_+}),
    \end{align*}
    where
    \begin{align*}
        \mathbbm{r}_{M_{\alpha, z} \to z_+} = \frac{1}{e^{\delta} + 1 +\mathbbm{r}_z (e^{\delta} -1)} \Big[ 2 e^{\frac{\delta}{2}} \mathbbm{r}_x, \;  2 e^{\frac{\delta}{2}} \mathbbm{r}_y, \; \mathbbm{r}_z (e^{\delta} +1) + e^{\delta} -1 \Big]
    \end{align*}
    and
    \begin{align*}
        \rho(\mathbbm{r})_{M_{\alpha, z} \to z_-} = \rho(\mathbbm{r}_{M_{\alpha, z} \to z_-}),
    \end{align*}
    where
    \begin{align*}
        \mathbbm{r}_{M_{\alpha, z} \to z_-} = \frac{1}{e^{\delta} + 1 - \mathbbm{r}_z (e^{\delta} -1)} \Big[ 2 e^{\frac{\delta}{2}} \mathbbm{r}_x, \;  2 e^{\frac{\delta}{2}} \mathbbm{r}_y, \; \mathbbm{r}_z (e^{\delta} +1) - (e^{\delta} -1)\Big].
    \end{align*}
\end{lemma}

\begin{proof}
    We may represent any state $\rho \in \mathcal{S}(\mathbb{C}^2)$ as $\rho = \rho(\mathbbm{r})$, where
    \begin{equation*}
        \rho(\mathbbm{r}) = \frac{1}{2} \left( \mathbbm{1}_{\mathbb{C}^2}+ \mathbbm{r} \cdot \sigma \right) = \frac{1}{2} \left( \mathbbm{1}_{\mathbb{C}^2}+ \mathbbm{r}_x \sigma_x + \mathbbm{r}_y \sigma_y + \mathbbm{r}_z \sigma_z \right).
    \end{equation*}
    $\sigma = [\sigma_x, \sigma_y, \sigma_z]^t$ are the Pauli-matrices and $\mathbbm{r} = [\mathbbm{r}_x, \mathbbm{r}_y, \mathbbm{r}_z]$ with $\norm{\mathbbm{r}}_2 \leq 1$ is the coordinate vector in the Pauli-basis representation. Note that we have
    \begin{equation}
    \label{eqn::Pauli_matrices_basis_rep}
        \sigma_x = \ket{{x_+}}\bra{{x_+}} - \ket{{x_-}}\bra{{x_-}}, \hspace{10pt} \sigma_y = \ket{{y_+}}\bra{{y_+}} - \ket{{y_-}}\bra{{y_-}}, \hspace{10pt} \sigma_z = \ket{{z_+}}\bra{{z_+}} - \ket{{z_-}}\bra{{z_-}}.
    \end{equation}
    Measuring $\rho(\mathbbm{r})$ using $M_{\alpha, z}$ yields the outcome $z_+$ with probability
    \begin{align*}
        \mathcal{P}_{\mathbbm{r}}\left( R_{M_{\alpha, z}} = z_+ \right) =& \Tr \left( \rho(\mathbbm{r}) M_{\alpha, z_+} M_{\alpha, z_+}^* \right) =  \Tr \left( \rho(\mathbbm{r}) \left( \frac{e^{\delta}}{e^{\delta}+1} \ket{{z_+}}\bra{{z_+}} + \frac{1}{e^{\delta}+1} \ket{{z_-}}\bra{{z_-}}\right) \right)
        \\
        =& \frac{1}{2} \frac{1}{e^{\delta}+1} \Big( e^{\delta} \bra{{z_+}}\mathbbm{1}_{\mathbb{C}^2}+ \mathbbm{r}_x \sigma_x + \mathbbm{r}_y \sigma_y + \mathbbm{r}_z \sigma_z \ket{{z_+}}
        \\
        & + \bra{{z_-}}\mathbbm{1}_{\mathbb{C}^2}+ \mathbbm{r}_x \sigma_x + \mathbbm{r}_y \sigma_y + \mathbbm{r}_z \sigma_z \ket{{z_-}} \Big).
    \end{align*}
    Now, using equations (\ref{eqn::Pauli_matrices_basis_rep}) and (\ref{eqn::mutually_unbiasedness}) we have
    \begin{align*}
        \bra{{z_+}}\mathbbm{1}_{\mathbb{C}^2}+ \mathbbm{r}_x \sigma_x + \mathbbm{r}_y \sigma_y + \mathbbm{r}_z \sigma_z \ket{{z_+}} &= 1 + \mathbbm{r}_z
        \intertext{and}
        \bra{{z_-}}\mathbbm{1}_{\mathbb{C}^2}+ \mathbbm{r}_x \sigma_x + \mathbbm{r}_y \sigma_y + \mathbbm{r}_z \sigma_z \ket{{z_-}} &= 1 - \mathbbm{r}_z.
    \end{align*}
    This gives
    \begin{align*}
        \mathcal{P}_{\mathbbm{r}}\left( R_{M_{\alpha, z}} = z_+ \right) &= \frac{1}{2} \frac{1}{e^{\delta}+1} \left( e^{\delta} (1 + \mathbbm{r}_z) + 1 - \mathbbm{r}_z \right) = \frac{1}{2} + \frac{1}{2} \frac{e^{\delta} - 1}{e^{\delta} + 1} \mathbbm{r}_z
        \intertext{and analogously}
        \mathcal{P}_{\mathbbm{r}}\left( R_{M_{\alpha, z}} = z_-\right) &= \frac{1}{2} \frac{1}{e^{\delta}+1} \left( e^{\delta} (1 - \mathbbm{r}_z) + 1 + \mathbbm{r}_z \right) = \frac{1}{2} - \frac{1}{2} \frac{e^{\delta} - 1}{e^{\delta} + 1} \mathbbm{r}_z.
    \end{align*}
    We will now consider what happens to the state after the measurement. When we obtain outcome $+1$ the state collapses to
    \begin{align}
        \rho(\mathbbm{r})_{M_{\alpha, z} \to z_+} &= \frac{1}{\mathcal{P}_{\mathbbm{r}}\left( R_{M_{\alpha, z}} = z_+ \right)} M_{\alpha, z_+} \rho(\mathbbm{r}) M_{\alpha, z_+}^* \notag
        \\
        &= \frac{1}{2} \frac{1}{\mathcal{P}_{\mathbbm{r}}\left( R_{M_{\alpha, z}} = z_+ \right)} M_{\alpha, z_+} \left( \mathbbm{1}_{\mathbb{C}^2}+ \mathbbm{r}_x \sigma_x + \mathbbm{r}_y \sigma_y + \mathbbm{r}_z \sigma_z \right)M_{\alpha, z_+} \notag
        \\
        \label{eqn::z_axis_+_post_measurement_state}
        &= \frac{1}{2} \frac{1}{\mathcal{P}_{\mathbbm{r}}\left( R_{M_{\alpha, z}} = z_+ \right)} \left( S_1 + \mathbbm{r}_x S_2 + \mathbbm{r}_y S_3 + \mathbbm{r}_z S_4 \right)
    \end{align}
    and we will calculate the terms $S_1,...,S_4$ separately. For $S_1$ and $S_4$ we have
    \begin{align*}
        S_1 &= M_{\alpha, z_+}M_{\alpha, z_+} = \frac{e^{\delta}}{e^{\delta} + 1} \ket{{z_+}}\bra{{z_+}} + \frac{1}{e^{\delta} + 1} \ket{{z_+}}\bra{{z_+}}
        \\
        S_4 &= M_{\alpha, z_+} \sigma_z M_{\alpha, z_+} = \frac{e^{\delta}}{e^{\delta}+1} \ket{{z_+}}\bra{{z_+}} - \frac{1}{e^{\delta}+1} \ket{{z_+}}\bra{{z_+}}.
    \end{align*}
    For $S_2$ we obtain
    \begin{alignat*}{2}
        S_2 =& M_{\alpha, z_+} && \sigma_x M_{\alpha, z_+} 
        \\
        =& M_{\alpha, z_+} &&\left( \ket{{x_+}}\bra{{x_+}} -  \ket{{x_-}}\bra{{x_-}}\right) M_{\alpha, z_+}
        \\
        =& \frac{1}{e^{\delta} + 1} && \Big[ \left( e^{\frac{\delta}{2}} \ket{{z_+}} \bra{{z_+}}\ket{{x_+}} + \ket{{z_-}} \bra{{z_-}}\ket{{x_+}} \right) \left( e^{\frac{\delta}{2}} \bra{{x_+}}\ket{{z_+}} \bra{{z_+}} + \bra{{x_+}}\ket{{z_-}} \bra{{z_-}} \right)
        \\
        & && - \left( e^{\frac{\delta}{2}} \ket{{z_+}} \bra{{z_+}}\ket{{x_-}} + \ket{{z_-}} \bra{{z_-}}\ket{{x_-}} \right) \left( e^{\frac{\delta}{2}} \bra{{x_-}}\ket{{z_+}} \bra{{z_+}} + \bra{{x_-}}\ket{{z_-}} \bra{{z_-}} \right) \Big]
        \\
        =& \frac{1}{e^{\delta} + 1} && \Big[ \left( \ket{{x_+}} + \left( e^{\frac{\delta}{2}} - 1 \right) \bra{{z_+}}\ket{{x_+}} \ket{{z_+}} \right) \left( \bra{{x_+}} + \left( e^{\frac{\delta}{2}} - 1 \right) \bra{{x_+}}\ket{{z_+}} \bra{{z_+}} \right)
        \\
        & && - \left( \ket{{x_-}} + \left( e^{\frac{\delta}{2}} - 1 \right) \bra{{z_+}}\ket{{x_-}} \ket{{z_+}} \right) \left( \bra{{x_-}} + \left( e^{\frac{\delta}{2}} - 1 \right) \bra{{x_-}}\ket{{z_+}} \bra{{z_+}} \right) \Big]
        \\
        =& \frac{1}{e^{\delta}+1} &&\Bigg[ \ket{{x_+}}\bra{{x_+}} + \left( e^{\frac{\delta}{2}} - 1 \right)^2 \left|\bra{{z_+}}\ket{{x_+}}\right|^2 \ket{{z_+}}\bra{{z_+}} 
        \\
        & && + \left( e^{\frac{\delta}{2}} - 1 \right) \left( \bra{{z_+}}\ket{{x_+}} \ket{{z_+}}\bra{{x_+}} + \bra{{x_+}}\ket{{z_+}} \ket{{x_+}}\bra{{z_+}} \right)
        \\
        & && - \Big( \ket{{x_-}}\bra{{x_-}} + \left( e^{\frac{\delta}{2}} - 1 \right)^2 \left|\bra{{z_+}}\ket{{x_-}}\right|^2 \ket{{z_+}}\bra{{z_+}} \Big)
        \\
        & && - \Big( \left( e^{\frac{\delta}{2}} - 1 \right) \left( \bra{{z_+}}\ket{{x_-}} \ket{{z_+}}\bra{{x_-}} + \bra{{x_-}}\ket{{z_+}} \ket{{x_-}}\bra{{z_+}} \right) \Big) \Bigg].
        \intertext{Now, by equation (\ref{eqn::mutually_unbiasedness}) we have $\left|\bra{{z_+}}\ket{{x_+}}\right|^2 = \frac{1}{2} = \left|\bra{{z_+}}\ket{{x_-}}\right|^2$ which gives}
        S_2 =& \frac{1}{e^{\delta}+1} &&\Big[ \ket{{x_+}}\bra{{x_+}} - \ket{{x_-}}\bra{{x_-}}
        \\
        & && \left( e^{\frac{\delta}{2}} - 1 \right) \Big( \bra{{x_+}}\ket{{z_+}} \ket{{x_+}} - \bra{{x_-}}\ket{{z_+}} \ket{{x_-}} \Big) \bra{{z_+}} + 
        \\
        & && \left( e^{\frac{\delta}{2}} - 1 \right)  \ket{{z_+}} \Big( \bra{{z_+}}\ket{{x_+}} \bra{{x_+}} - \bra{{z_+}}\ket{{x_-}} \bra{{x_-}} \Big).
        \intertext{Straightforeward calculations show that $\bra{{x_+}}\ket{{z_+}} \ket{{x_+}} - \bra{{x_-}}\ket{{z_+}} \ket{{x_-}} = \ket{{z_-}}$ such that we can write}
        S_2 =& \frac{1}{e^{\delta} + 1} && \Big[ \ket{{x_+}}\bra{{x_+}} - \ket{{x_-}}\bra{{x_-}} + \left( e^{\frac{\delta}{2}} - 1 \right) \left( \ket{{z_-}}\bra{{z_+}} + \ket{{z_+}}\bra{{z_-}} \right) \Big]
        \\
        =& \frac{1}{e^{\delta}+1} &&\Big[ \sigma_x + \left( e^{\frac{\delta}{2}} -1 \right) \sigma_x] 
        \\
        =& \frac{e^{\frac{\delta}{2}}}{e^{\delta} + 1} && \sigma_x 
        \end{alignat*}
    since we have \begin{align*}
        (\ket{{z_-}}\bra{{z_+}} + \ket{{z_+}}\bra{{z_-}})\ket{{x_+}} &= \ket{{x_+}}
        \\
        (\ket{{z_-}}\bra{{z_+}} + \ket{{z_+}}\bra{{z_-}})\ket{{x_-}} &= -\ket{{x_-}}
    \end{align*}
    which shows that $\ket{{z_-}}\bra{{z_+}} + \ket{{z_+}}\bra{{z_-}} = \sigma_x$. \\
    
    Lets now turn to $S_3$. Here we have (similarly to $S_2$)
    \begin{alignat*}{2}
        S_3 =& M_{\alpha, z_+} && \sigma_y M_{\alpha, z_+} 
        \\
        =& M_{\alpha, z_+} &&\left( \ket{{y_+}}\bra{{y_+}} -  \ket{{y_-}}\bra{{y_-}}\right) M_{\alpha, z_+}
        \\
        =& \frac{1}{e^{\delta}+1} &&\Big[ \ket{{y_+}}\bra{{y_+}} - \ket{{y_-}}\bra{{y_-}}
        \\
        & && \left( e^{\frac{\delta}{2}} - 1 \right) \Big( \bra{{y_+}}\ket{{z_+}} \ket{{y_+}} - \bra{{y_-}}\ket{{z_+}} \ket{{y_-}} \Big) \bra{{z_+}} + 
        \\
        & && \left( e^{\frac{\delta}{2}} - 1 \right)  \ket{{z_+}} \Big( \bra{{z_+}}\ket{{y_+}} \bra{{y_+}} - \bra{{z_+}}\ket{{y_-}} \bra{{y_-}} \Big).
        \intertext{Here we can again calculate $\bra{{y_+}}\ket{{z_+}} \ket{{y_+}} - \bra{{y_-}}\ket{{z_+}} \ket{{y_-}} = i \ket{{z_-}}$ such that we obtain}
        S_3 =& \frac{1}{e^{\delta} + 1} && \Big[ \ket{{y_+}}\bra{{y_+}} - \ket{{y_-}}\bra{{y_-}} + \left( e^{\frac{\delta}{2}} - 1 \right) \left( i \ket{{z_-}}\bra{{z_+}} - i \ket{{z_+}}\bra{{z_-}} \right) \Big]
        \\
        =& \frac{1}{e^{\delta}+1} &&\Big[ \sigma_y + \left( e^{\frac{\delta}{2}} -1 \right) \sigma_y] 
        \\
        =& \frac{e^{\frac{\delta}{2}}}{e^{\delta} + 1} && \sigma_y
    \end{alignat*}
    since we have
    \begin{align*}
        \left(i \ket{{z_-}}\bra{{z_+}} - i \ket{{z_+}}\bra{{z_-}}\right)\ket{{y_+}} &= \ket{{y_+}}
        \\
        \left(i \ket{{z_-}}\bra{{z_+}} - i \ket{{z_+}}\bra{{z_-}}\right)\ket{{y_-}} &= -\ket{{y_-}}
    \end{align*}
    which shows $i \ket{{z_-}}\bra{{z_+}} - i \ket{{z_+}}\bra{{z_-}} = \sigma_y$.\\

    Returning to equation (\ref{eqn::z_axis_+_post_measurement_state}), we can now write this as
    \begin{alignat*}{2}
        &\rho(\mathbbm{r})_{M_{\alpha, z} \to z_+} &&
        \\
        =& \frac{1}{2} \frac{1}{\mathcal{P}_{\mathbbm{r}}\left( R_{M_{\alpha, z}} = z_+ \right)} \Big( S_1 + \mathbbm{r}_x && S_2 + \mathbbm{r}_y S_3 + \mathbbm{r}_z S_4 \Big)
        \\
        =& \frac{1}{2}\frac{1}{\mathcal{P}_{\mathbbm{r}}\left( R_{M_{\alpha, z}} = z_+ \right)} \frac{1}{e^{\delta} + 1}
        \Big[ e^{\delta} && \ket{{z_+}}\bra{{z_+}} + \ket{{z_-}}\bra{{z_-}}
        \\
        & && + e^{\frac{\delta}{2}} \mathbbm{r}_x \sigma_x + e^{\frac{\delta}{2}} \mathbbm{r}_y \sigma_y + \mathbbm{r}_z e^{\delta} \ket{{z_+}}\bra{{z_+}} - \mathbbm{r}_z \ket{{z_-}}\bra{{z_-}} \Big]
        \\
        =& \frac{1}{2 \, \mathcal{P}_{\mathbbm{r}}\left( R_{M_{\alpha, z}} = z_+ \right) (e^{\delta} +1 )} \Big[ && \frac{1}{2} \left( e^{\delta} + 1 + \mathbbm{r}_z \left( e^{\delta} - 1 \right) \right) \mathbbm{1}_{\mathbb{C}^2}
        \\
        & && + e{^\frac{\delta}{2}} \mathbbm{r}_x \sigma_x + e{^\frac{\delta}{2}} \mathbbm{r}_y \sigma_y + \frac{1}{2} \left( e^{\delta} - 1 + \mathbbm{r}_z \left( e^{\delta} +1 \right) \right) \sigma_z \Big]
        \\
        =& \frac{1}{2}\Big[ \mathbbm{1}_{\mathbb{C}^2} + \frac{2 e^{\frac{\delta}{2}} \mathbbm{r}_x}{e^{\delta} + 1 +\mathbbm{r}_z (e^{\delta} -1)} && \sigma_x + \frac{2 e^{\frac{\delta}{2}} \mathbbm{r}_y}{e^{\delta} + 1 +\mathbbm{r}_z (e^{\delta} -1)} \sigma_y + \frac{e^{\delta} -1 + \mathbbm{r}_z (e^{\delta} +1)}{e^{\delta} + 1 + \mathbbm{r}_z(e^{\delta} -1)} \sigma_z \Big]
        \\
        =& \rho(\mathbbm{r}_{M_{z,\delta} \to z_+}), &&
    \end{alignat*}
    where
    \begin{align*}
        \mathbbm{r}_{M_{z,\delta} \to z_+} &= \Big[ \frac{2 e^{\frac{\delta}{2}} \mathbbm{r}_x}{e^{\delta} + 1 +\mathbbm{r}_z (e^{\delta} -1)}, \frac{2 e^{\frac{\delta}{2}} \mathbbm{r}_y}{e^{\delta} + 1 +\mathbbm{r}_z (e^{\delta} -1)}, \frac{\mathbbm{r}_z (e^{\delta} +1) + e^{\delta} -1}{e^{\delta} + 1 + \mathbbm{r}_z(e^{\delta} -1)} \Big]
        \\
        &= \frac{1}{e^{\delta} + 1 +\mathbbm{r}_z (e^{\delta} -1)} \Big[ 2 e^{\frac{\delta}{2}} \mathbbm{r}_x, \;  2 e^{\frac{\delta}{2}} \mathbbm{r}_y, \; \mathbbm{r}_z (e^{\delta} +1) + e^{\delta} -1 \Big].
    \end{align*}

    Very similar calculations give the post measurement state when the outcome $z_-$ is obtained. There we have
    \begin{align*}
        &\rho(\mathbbm{r})_{M_{\alpha, z} \to z_-}
        \\
        =& \frac{1}{2}\Big[ \mathbbm{1}_{\mathbb{C}^2} + \frac{2 e^{\frac{\delta}{2}} \mathbbm{r}_x}{e^{\delta} + 1 - \mathbbm{r}_z (e^{\delta} -1)} \sigma_x + \frac{2 e^{\frac{\delta}{2}} \mathbbm{r}_y}{e^{\delta} + 1 - \mathbbm{r}_z (e^{\delta} -1)} \sigma_y + \frac{1 - e^{\delta} + \mathbbm{r}_z (e^{\delta} + 1)}{e^{\delta} + 1 - \mathbbm{r}_z(e^{\delta} -1)} \sigma_z \Big]
        \\
        =& \rho(\mathbbm{r}_{M_{z,\delta} \to z_-}),
    \end{align*}
    where
    \begin{align*}
        \mathbbm{r}_{M_{z,\delta} \to z_-} &= \Big[ \frac{2 e^{\frac{\delta}{2}} \mathbbm{r}_x}{e^{\delta} + 1 -\mathbbm{r}_z (e^{\delta} -1)}, \frac{2 e^{\frac{\delta}{2}} \mathbbm{r}_y}{e^{\delta} + 1 -\mathbbm{r}_z (e^{\delta} -1)}, \frac{\mathbbm{r}_z (e^{\delta} +1) - (e^{\delta} -1)}{e^{\delta} + 1 - \mathbbm{r}_z(e^{\delta} -1)} \Big]
        \\
        &= \frac{1}{e^{\delta} + 1 - \mathbbm{r}_z (e^{\delta} -1)} \Big[ 2 e^{\frac{\delta}{2}} \mathbbm{r}_x, \;  2 e^{\frac{\delta}{2}} \mathbbm{r}_y, \; \mathbbm{r}_z (e^{\delta} +1) - (e^{\delta} -1)\Big].
    \end{align*}
\end{proof}

\subsection{Additional proofs}
\label{sec::additonal_proofs}

\begin{blank}{\textbf{Restatement of Lemma \ref{lem::Help1}.}}
    Let $M$ be a quantum measurement with measurement operators $M_y$, where $E_y := M_y^* M_y$. Furthermore, let $\rho_i \in \mathcal{S}(\mathbb{C}^d)$ be quantum states. We set $$
    p_i^{R_M}(y) := \mathcal{P}_i(R_M = y) = \Tr(\rho_i E_y).
    $$ 
    Then:
    \begin{enumerate}
        \item[(i)] For all $y$ it holds: $\lambda_{min}(E_y) \leq p_i^{R_M}(y) \leq \lambda_{max}(E_y)$,

        \item[(ii)] For all $y$ it holds: $\left| p_{1}^{R_M}(y) - p_{2}^{R_M}(y) \right| \leq \left( \lambda_{max}(E_y) - \lambda_{min}(E_y) \right) \norm{\rho_1 - \rho_2}_{Tr}$ .
    \end{enumerate}
\end{blank}

\begin{proof}
    Let us start by showing that for two hermitian positive semi-definite matrices $A$ and $E$ we have
    \begin{equation}
        \label{eqn::lemHelp2}
        \lambda_{min}(E) \Tr(A) \leq \Tr(AE) \leq \lambda_{max}(E) \Tr(A).
    \end{equation}
    To see this, we note that for any hermitian and positive semi-definite matrices $A$ and $B$, even though $AB$ is not necessarily hermitian or positive, $\Tr(AB)$ is still positive. To see this, write
    \begin{equation*}
        \Tr(AB) = \Tr( \sqrt{A} B \sqrt{A}) = \sum_{i = 1}^{d} \bra{e_i} \sqrt{A} B \sqrt{A} \ket{e_i} \geq 0
    \end{equation*}
    for an ONB ($\ket{e_i}$) of $\mathbb{C}^d$. Evidently, we have 
    \begin{equation*}
        B_1 = E - \lambda_{min}(E) \mathbbm{1}_{\mathbb{C}^d} \geq 0 \hspace{20pt} \text{and} \hspace{20pt} B_2 = \lambda_{max}(E) \mathbbm{1}_{\mathbb{C}^d} - E \geq 0.
    \end{equation*}
    From this, we obtain
    \begin{align*}
        \Tr(AE) - \lambda_{min}(E) \Tr(A) &= \Tr \left( A (E - \lambda_{min}(E) \mathbbm{1}_{\mathbb{C}^d}) \right) = \Tr(AB_1) \geq 0
        \intertext{and}
        \lambda_{max}(E) \Tr(A) - \Tr(AE) &= \Tr \left( A (\lambda_{max}(E) \mathbbm{1}_{\mathbb{C}^d} - E) \right) = \Tr(AB_2) \geq 0.
    \end{align*}
    In order to show (i) we now note that by the definition of a quantum measurement we have
    \begin{align*}
        p_{i}^{R_M}(y) = \Tr(\rho_i E_y).
    \end{align*}
    The assertion follows from the inequalities in \eqref{eqn::lemHelp2} and the fact that $\Tr(\rho_i) = 1$. To show (ii) we note that by definition, we have:
    \begin{align*}
        p_{1}^{R_M}(y) - p_{2}^{R_M}(y) &= \Tr(\rho_1 E_y) - \Tr(\rho_2 E_y)
        \\
        &= \Tr(E_y (\rho_1 - \rho_2)) 
        \\
        &= \Tr(E_y (\rho_1 - \rho_2)_+) - \Tr(E_y (\rho_1 - \rho_2)_-)
        \\
        &\leq \lambda_{max}(E_y) \Tr((\rho_1 - \rho_2)_+) - \lambda_{min}(E_y) \Tr((\rho_1 - \rho_2)_-)
        \\
        &= \left( \lambda_{max}(E_y) - \lambda_{min}(E_y) \right) \Tr\left( (\rho_1 - \rho_2)_+ \right)
        \\
        &= \left( \lambda_{max}(E_y) - \lambda_{min}(E_y) \right) \norm{\rho_1 - \rho_2}_{Tr}
    \end{align*}
    where we used $\rho_1 - \rho_2 = (\rho_1 - \rho_2)_+ - (\rho_1 - \rho_2)_-$ together with inequalities~\eqref{eqn::lemHelp2} and 
    \begin{equation*}
        0 = \Tr(\rho_1 - \rho_2) = \Tr\left( (\rho_1 - \rho_2)_+ \right) - \Tr\left( (\rho_1 - \rho_2)_- \right)
    \end{equation*}
    to see with $|\rho_1 - \rho_2| = (\rho_1 - \rho_2)_+ + (\rho_1 - \rho_2)_-$that 
    \begin{equation*}
        \Tr\left( (\rho_1 - \rho_2)_+ \right) = \norm{\rho_1 - \rho_2}_{Tr}.
    \end{equation*}
        The symmetry of the term yields the result.
\end{proof}

\begin{blank}{\textbf{Restatement of Lemma \ref{lemma::3}.}}
    If $M$ is an $\alpha$-gentle measurement on the set of pure states $\mathcal{S}_{pure}(\mathbb{C}^d)$ such that the $M_y$ are positive semi-definite, then $M$ is already $\alpha$-gentle on the whole set of states $\mathcal{S}(\mathbb{C}^d)$.
\end{blank}

\begin{proof}
    Let us assume $M_y$ to be positive semi-definite and $\rho = \sum_{i} \mu_i \ket{\psi_i}\bra{\psi_i} $ in $\mathcal{S}(\mathbb{C}^d)$ be a (possibly mixed) state. By the Schrödinger-HJW Theorem \cite{Kirkpatrick2006}, we can purify the state on $\mathbb{C}^d \otimes \mathbb{C}^d$ in the following way. Define $\ket{\Psi} = \sum_{i} \sqrt{\mu_i} \ket{\psi_i} \otimes \ket{\psi_i}$. Then we have $\rho = \Tr_2\left( \ket{\Psi}\bra{\Psi} \right)$, where $\Tr_2$ denotes the partial trace over the second copy of $\mathbb{C}^d$ and $M_y$ is extended to $M_y \otimes \text{Id}$. Since the trace-norm is contractive under quantum channels we have
    \begin{align*}
        \norm{\rho - \rho_{M \to y}}_{Tr} &= \norm{\Tr_2\left( \ket{\Psi}\bra{\Psi} \right) - \Tr_2\left( \ket{\Psi}\bra{\Psi} \right)_{M \to y}}_{Tr} 
        \\
        &\leq \norm{\ket{\Psi}\bra{\Psi} -  \ket{\Psi}\bra{\Psi}_{M \otimes \text{Id} \to y}}_{Tr}
        \\
        &= 1 - \frac{\left| \bra{\Psi} (M_y \otimes \text{Id}) \ket{\Psi} \right|^2}{\bra{\Psi} (M_y^* M_y \otimes \text{Id}) \ket{\Psi}}.
    \end{align*}
    By assumption, we can write
    \begin{equation*}
        M_y = \sum_{i = 1}^{d} \mu_i \ket{v_i}\bra{v_i}
    \end{equation*}
    for an orthonormal basis $\ket{v_i}$ of $\mathcal{H}$. Any pure state $\ket{\psi} \in \mathcal{H}$ can therefore be written as
    \begin{equation*}
        \ket{\psi} = \sum_{i = 1}^{d} \sqrt{\nu_i} \ket{v_i}, \hspace{10pt} \text{where } \sum_{i = 0}^{d-1} |\nu_i| = 1.
    \end{equation*}
    Straightforward calculations now give
    \begin{align*}
        \bra{\psi}M_y\ket{\psi} = \sum_{i = 1}^{d} \mu_i |\nu_i| \hspace{20pt} \text{and} \hspace{20pt}
        \bra{\psi}M_y^2\ket{\psi} = \sum_{i = 1}^{d} \mu_i^2 |\nu_i|.
    \end{align*}
    Now, $M_y \otimes \text{Id}$ has the same eigenvalues as $M_y$ and the vectors $(\ket{v_i} \otimes \ket{v_j})_{i,j \in \{ 0,...,d-1\}}$ form an ONB of $\mathcal{H} \otimes \mathcal{H}$. Any pure state $\ket{\Psi} \in \mathcal{H} \otimes \mathcal{H}$ can therefore be written as
    \begin{equation*}
        \ket{\Psi} = \sum_{i,j = 1}^{d} \sqrt{\nu_{i,j}} \ket{v_i} \otimes \ket{v_j} , \hspace{10pt} \text{where } \sum_{i,j = 1}^{d} |\nu_{i,j}| = 1.
    \end{equation*}
    We can now again calculate
    \begin{align*}
        \bra{\Psi}(M_y \otimes \text{Id}) \ket{\Psi} = \sum_{i = 1}^{d} \mu_i \sum_{j = 0}^{d-1} |\nu_{i,j}| =: \sum_{i = 1}^{d} \mu_i |\Tilde{\nu_i}|
        \intertext{and}
        \bra{\Psi}(M_y^2 \otimes \text{Id}) \ket{\Psi} = \sum_{i = 1}^{d} \mu_i^2 \sum_{j = 1}^{d} |\nu_{i,j}| =: \sum_{i = 1}^{d} \mu_i^2 |\Tilde{\nu_i}|
    \end{align*}
    where we defined $|\Tilde{\nu_i}| := \sum_{j = 1}^{d} |\nu_{i,j}|$. We can see that for every pure state 
    \begin{equation*}
        \ket{\Psi} = \sum_{i,j = 1}^{d} \sqrt{\nu_{i,j}} \ket{v_i} \otimes \ket{v_j}
    \end{equation*}
    we can define a pure state 
    \begin{equation*}
        \ket{\psi} = \sum_{i = 1}^d \Tilde{\nu}_i \ket{v_i}, \hspace{20pt} \text{where } \Tilde{\nu_i} := \sqrt{\sum_{j = 1}^{d} |\nu_{i,j}|}
    \end{equation*}
    such that 
    \begin{align*}
        \frac{\left| \bra{\Psi} (M_y \otimes \text{Id}) \ket{\Psi} \right|^2}{\bra{\Psi} M_y^* M_y \otimes \text{Id}) \ket{\Psi}} = \frac{\left| \bra{\psi} M_y \ket{\psi} \right|^2}{\bra{\psi} M_y^* M_y \ket{\psi}}
    \end{align*}
    which shows that $\alpha$-gentleness of $M$ on the pure states on $\mathcal{S}_{pure}(\mathbb{C}^d)$ implies $\alpha$-gentleness of $M$ on $\mathcal{S}(\mathbb{C}^d)$.
\end{proof}

\newpage
\nocite{*}

\bibliographystyle{alpha} 
\bibliography{Qubitsbiblio}

\end{document}